\documentclass{article}

\author{Pavel G. Naumov \and Jia Tao}
\title{The Budget-Constrained Functional Dependency}

\usepackage{amssymb}
\usepackage{amsmath}
\usepackage{wasysym,natbib}
\usepackage{graphics}
\newtheorem{definition}{Definition}
\newtheorem{theorem}{Theorem}
\newtheorem{lemma}{Lemma}
\newtheorem{claim}{Claim}
\newtheorem{corollary}{Corollary}
\newtheorem{proposition}{Proposition}
\newenvironment{proof}{\noindent{\sf Proof.}}{\hfill $\boxtimes$\linebreak}

\renewcommand{\phi}{\varphi}
\newcommand{\<}{\langle}
\renewcommand{\>}{\rangle}

\newcommand{\Section}{Section}
\newcommand{\idot}{. }

\begin{document}

\maketitle

\begin{abstract}
Armstrong's axioms of functional dependency form a well-known logical system that captures properties of functional dependencies between sets of database attributes. This article assumes that there are costs associated with attributes and proposes an extension of Armstrong's system for reasoning about budget-constrained functional dependencies in such a setting.

The main technical result of this article is the completeness theorem for the proposed logical system. Although the proposed axioms are obtained by just adding cost subscript to the original Armstrong's axioms, the proof of the completeness for the proposed system is significantly more complicated than that for the Armstrong's system.
\end{abstract}

\section{Introduction}

In dependency theory, functional dependencies are often used not only as a description of the data, but also as a semantic constraints for database designs, see \cite{v85}.
\cite{a74} introduced a system of three axioms describing the properties of functional dependencies between sets of attributes in a database.  
The applicability of these axioms goes far beyond the domain of databases. They describe the properties of functional dependency between any two sets of values. For example, knowing sides $a$ and $b$ of a triangle and the angle $\gamma$ between them, one can determine the third side $c$ and two other angles $\alpha$ and $\beta$. We write this as $a,b,\gamma\rhd c,\alpha, \beta$. Yet, knowing two sides of the triangle and the angle {\em not} between them, one cannot determine the remaining site and angles:  $\neg(a,b,\alpha\rhd c,\beta,\gamma)$. 

The property  $a,b,\gamma\rhd c,\alpha, \beta$ is valid when $a,b$ and $c$ are three sides of a triangle and $\alpha$, $\beta$, $\gamma$ are the angles opposite to these sides respectively. However, it may not be valid under some other interpretation of variables $a,b,c,\alpha, \beta$, and $\gamma$. For example, it is not valid if $a,b$ and $c$ are three sides of a pentagon and $\alpha$, $\beta$, $\gamma$ are the opposite angles. Armstrong's axioms capture the most general properties of functional dependencies that are valid in all settings. These axioms are:
\begin{enumerate}
\addtolength{\itemindent}{2mm}
\item[(A1)] {\em Reflexivity}: $A\rhd B$, if $B\subseteq A$,
\item[(A2)] {\em Augmentation}: $A\rhd B \rightarrow A,C\rhd B,C$,
\item[(A3)] {\em Transitivity}: $A\rhd B \rightarrow (B\rhd C \rightarrow A\rhd C)$,
\end{enumerate}
where $A,B$ denotes the union of sets $A$ and $B$. Armstrong proved the soundness and the completeness of this logical system with respect to a database semantics.
The above axioms became known in database literature as Armstrong's axioms, see \citet[p.~81]{guw09}. Beeri, Fagin, and Howard~\cite{bfh77} suggested a variation of Armstrong's axioms that describes properties of multi-valued dependency. \cite{v07} proposed a first order version of these principles. \cite{nn14jpl} developed a similar set of axioms for what they called the {\em rationally} functional dependency.

There have been two different approaches to extending Armstrong's axioms to handle approximate reasoning. \cite{bv06dsaa} described a complete logical system that formally captures the relation that {\em approximate values of attributes in set $A$ functionally determine approximate values of attributes in set $B$}. In his upcoming work, \cite{v14arxiv} considered the relation {\em attributes in set $A$ determine attributes in set $B$ with exception of $p$ fraction of cases}. We denote this relation by $A\rhd_p B$. For example, $A\rhd_{0.05} B$ means that attributes in set $A$ determine attributes in set $B$ in all but 5\% of the cases. \cite{v14arxiv} proposed a complete axiomatic system for this relation, which is based on the following principles:

\begin{enumerate}
\item Reflexivity: $A\rhd_0 B$, where $B\subseteq A$,
\item Totality: $A\rhd_1 B$,
\item Weakening: $A\rhd_p C,D\to A,B\rhd_p C$,
\item Augmentation: $A\rhd_p B \rightarrow A,C\rhd_p B,C$,
\item Transitivity: $A\rhd_p B\to(B\rhd_q C\to A\rhd_{p+q}C)$, where $p+q\le 1$,
\item Monotonicity: $A\rhd_p B\to A\rhd_q B$, where $p\le q$.
\end{enumerate}

Note that {V{\"a}{\"a}n{\"a}nen}'s relation $A\rhd_p B$ is exactly the original Armstrong's functional dependency relation when $p=0$. In the case of an arbitrary $p$, relation $A\rhd_p B$ could be considered as a ``weaker" form of functional dependency, which might hold even in the cases where the functional dependency does not hold. 

In this article we propose another interpretation of atomic predicate $A\rhd_p B$, that we call {\em the budget-constrained dependency}. Just like {V{\"a}{\"a}n{\"a}nen}'s approximate dependency, the budget-constrained dependency is a weaker form of the original Armstrong's functional dependency relation. Intuitively, there is a budget-constrained dependency $A\rhd_p B$ when just a few new attributes could be added to the set $A$ in such a way that the extended set functionally determines the set $B$. We use parameter $p$ to formally specify what the phrase ``a few new attributes" means. Namely, we assume that a non-negative cost is assigned to each attribute and that $A\rhd_p B$ means that there is a way to add several attributes with a total cost no more than $p$ to set $A$ in such a way that the extended set of attributes functionally determines all attributes in set $B$.  In this article we introduce a complete logical system for the budget-constrained dependency which is based on the following three principles that generalize Armstrong's axioms:
\begin{enumerate}
\item {\em Reflexivity}: $A\rhd_p B$, if $B\subseteq A$,
\item {\em Augmentation}: $A\rhd_p B \rightarrow A,C\rhd_p B,C$,
\item {\em Transitivity}: $A\rhd_p B \rightarrow (B\rhd_q C \rightarrow A\rhd_{p+q} C)$.
\end{enumerate}

We call our framework ``the budget-constrained dependency" because its most natural application is in the setting where all data is potentially accessible to an agent at a cost. In this setting $A\rhd_p B$ means that if an agent already knows the values of attributes in set $A$, then she can determine the values of attributes in set $B$ at a cost $p$. One example of such a setting is fees associated with information access: criminal background check fees, court records obtaining fees, etc. Another example is geological explorations, where learning about deposits of mineral resources often requires costly drilling. Although it is convenient to think about a budget constraint as a financial one, a budget constraint can also refer to a limit on time, space, or some other resource.

Although Reflexivity (A1), Augmentation (A2), and Transitivity (A3) are usually called Armstrong's ``axioms", technically they can be formalized either as inference rules or as axioms. In the former case the language of the system consists only of the atomic predicate of the form $A\rhd_p B$. In the latter case the language consists of all Boolean combinations of such predicates and the formal logical system also includes propositional tautologies and the Modus Ponens inference rule.  \cite{a74} original paper proves the completeness of (A1), (A2), and (A3) as inference rules. However, his proof can be easily modified to prove the completeness of the corresponding system of axioms, as in \cite{hn14wollic}.

\cite{bv06dsaa} and \cite{v14arxiv} also proved the completeness theorems for logical systems consisting of inference rules. Different from their approach, in this article we treat our version of Armstrong's axioms as propositional axioms. Since inference rule
$$
\dfrac{\phi_1,\phi_2,\dots, \phi_n}{\psi}
$$
could be interpreted as a propositional formula
$$
\phi_1\to(\phi_2\to\dots (\phi_n\to\psi)\dots),
$$
an axiom-based system is syntactically richer than the corresponding rule-based system. Thus, the completeness result for an axiom-based system obtained in this article is stronger than a potential claim of the completeness result for the corresponding rule-based system.  

The main result of this article is the completeness theorem for our logical system. Since the system essentially consists of three Armstrong's axioms labeled with budget constrains, one might expect the proof of the completeness to be a straightforward modification of the proof of the completeness for the original Armstrong's system. Surprisingly, the proof of the completeness in our case requires a much more sophisticated argument. 

Next, we explain the reason for this unexpected complexity. The completeness theorem states that any unprovable formula is not satisfied in at least one model. Thus, to show the completeness we need to be able to construct a model (or a ``counterexample") for each unprovable formula. 
For instance, consider the formula $a\rhd b\to b\rhd a$ in the language without budget constraints. To construct a counterexample for this formula we need to describe a model in which attribute $a$ functionally determines attribute $b$ but not vice versa. To describe such a model, one can think of attributes $a$ and $b$ as paper folders that are used to store copies of certain documents. Specifically, consider a model in which folder $a$ always stores a copy of document $X$ and folder $b$ is always kept empty, see Figure~\ref{Folders-1 figure}. In this model, based on the content of folder $a$ one can always vacuously recover the content of empty folder $b$. At the same time, based on the content of empty folder $b$ one cannot recover the content of folder $a$. Thus, formula $a\rhd b\to b\rhd a$ is false in this model.

\begin{figure}[ht]
\begin{center}
\vspace{0mm}
\scalebox{.5}{\includegraphics{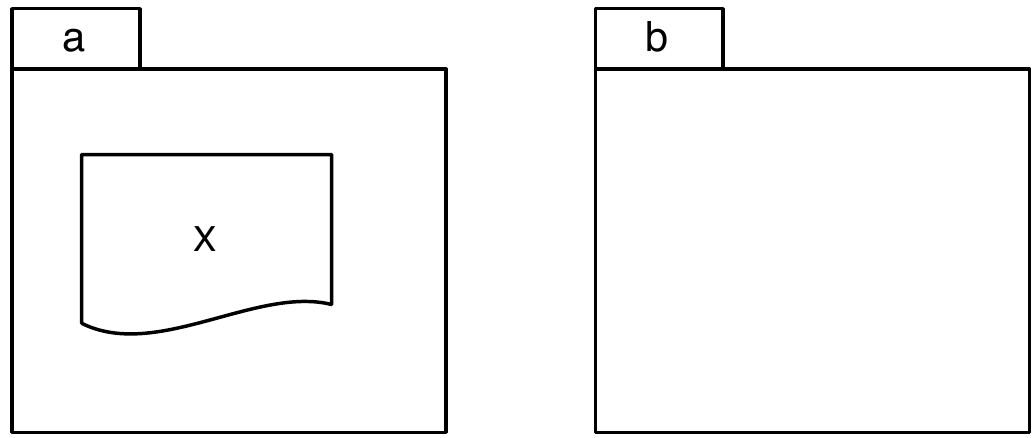}}
\vspace{0mm}
\footnotesize
\caption{Formula $a\rhd b$ is true, but formula $b\rhd a$ is false.}\label{Folders-1 figure}
\vspace{0cm}
\end{center}
\vspace{0mm}
\end{figure}

To construct a counterexample for formula $a\rhd b \vee b\rhd a$, one can consider a model in which folders $a$ and $b$ containing copies of two different (and unrelated to each other) documents $X$ and $Y$ respectively, see Figure~\ref{Folders-2 figure}. 

\begin{figure}[ht]
\begin{center}
\vspace{0mm}
\scalebox{.5}{\includegraphics{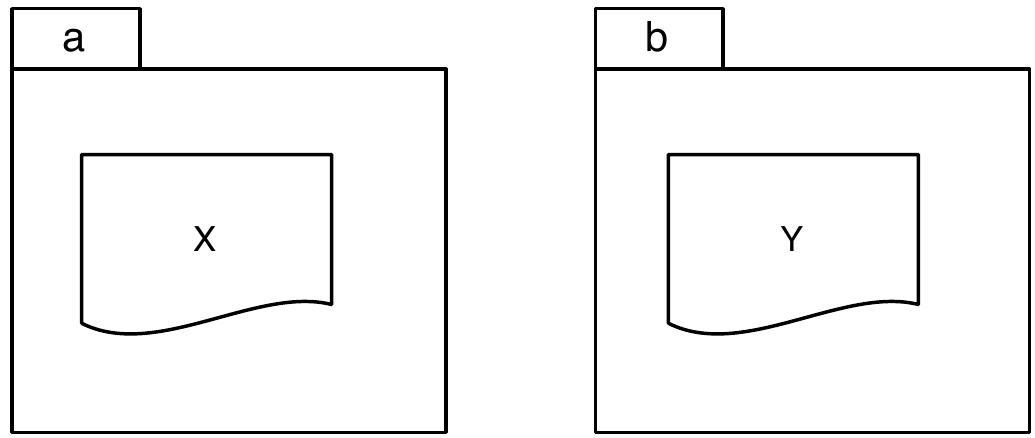}}
\vspace{0mm}
\footnotesize
\caption{Formulas $a\rhd b$ and $b\rhd a$ are both false.}\label{Folders-2 figure}
\vspace{0cm}
\end{center}
\vspace{0mm}
\end{figure}

To construct counterexamples for more complicated formulas, one can consider models with multiple folders containing copies of multiple documents. An example of such a model is depicted in Figure~\ref{Folders-3 figure}. In this model $a,b\rhd c$ is true because anyone with access to folders $a$ and $b$ knows the content of folder $c$. The folder/document model informally described here is sufficiently general to create a counterexample for each formula unprovable from Armstrong's axioms. In fact, the original Armstrong's proof of the completeness for his rule-based system and the proof of the completeness for the corresponding axiom-based system (\cite{hn14wollic}) could be viewed as formalizations of this folder/document construction. 

\begin{figure}[ht]
\begin{center}
\vspace{0mm}
\scalebox{.5}{\includegraphics{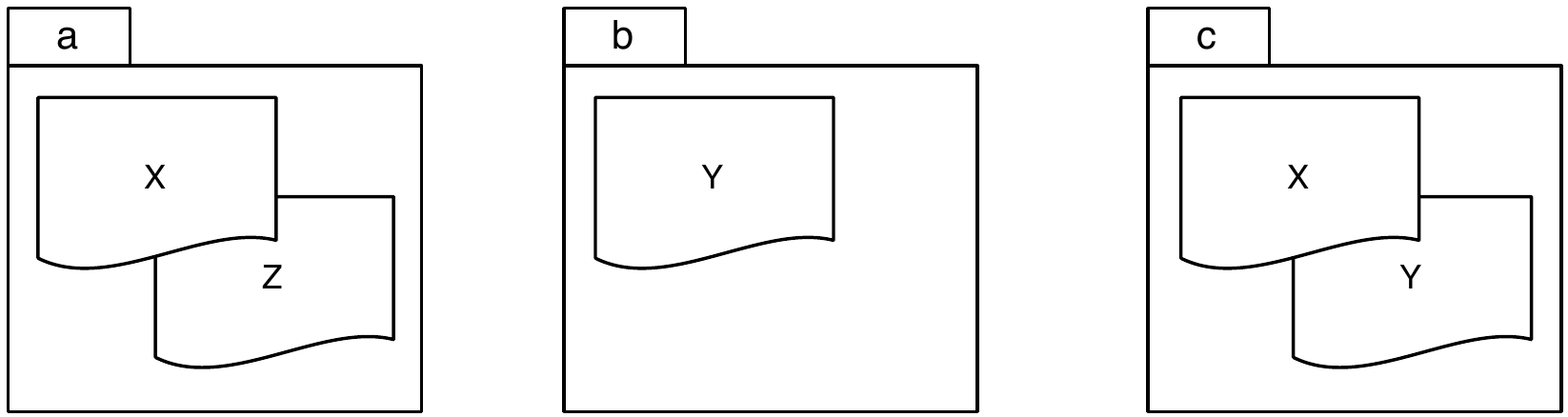}}
\vspace{0mm}
\footnotesize
\caption{Formula $a,b\rhd c$ is true.}\label{Folders-3 figure}
\vspace{0cm}
\end{center}
\vspace{0mm}
\end{figure}

The situation becomes significantly more complicated once the cost of information is added to the language. Let us start with a very simple example. If we want to construct a counterexample for formula $a\rhd_4 b$, then we can consider a model depicted in Figure~\ref{Folders-4 figure} with two folders: $a$ and $b$, priced at $\$3$ and $\$5$, respectively. The first of these folders is always empty and the second contains a copy of the document $X$. It is clear that in this model anyone
 who knows the content of folder $a$ still needs to spend $\$5$ to learn the content of folder $b$. Thus, budget-constrained dependency $a\rhd_p b$ is not satisfied in this model for each $p< 5$.

\begin{figure}[ht]
\begin{center}
\vspace{0mm}
\scalebox{.5}{\includegraphics{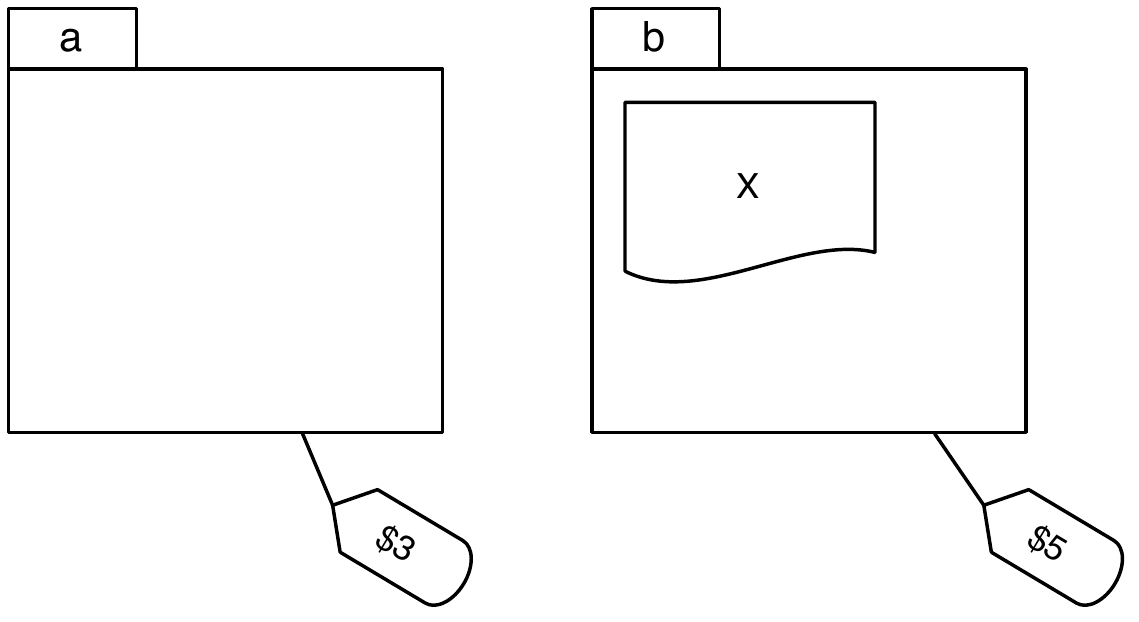}}
\vspace{0mm}
\footnotesize
\caption{Formula $a\rhd_4 b$ is false.}\label{Folders-4 figure}
\vspace{0cm}
\end{center}
\vspace{0mm}
\end{figure}

Let us now consider a more interesting example. Suppose that we want to construct a counterexample for formula $a\rhd_4 b\to\varnothing\rhd_4 b$. That is, we want to construct a model where anyone who knows the content of folder $a$ can reconstruct the content of folder $b$ after spending at most $\$4$. Yet, the same can not be done without access to folder $a$. To construct such a model we use the cryptographic tool called one-time encryption pad. Our model consists of three folders $a$, $b$, and $c$ priced at $\$3$, $\$5$, and $\$4$, respectively, see Figure~\ref{Folders-5 figure}. Let folder $b$ contain a copy of a document $X$, folder $c$ contain encryption an pad $P$, and folder $a$ contain the encrypted version of the document. In this model, $\varnothing\rhd_4 b$ is false because $\$4$ buys either access to the encryption pad in folder $c$ or access to the encrypted text in folder $a$, but not both. However,  formula $a\rhd_4 b$ is true in the same model because anyone who knows encrypted text $Encrypt(X,P)$ can spend $\$4$ on pad $P$, decode message $X$, and thus, learn the content of folder $b$.

\begin{figure}[ht]
\begin{center}
\vspace{0mm}
\scalebox{.5}{\includegraphics{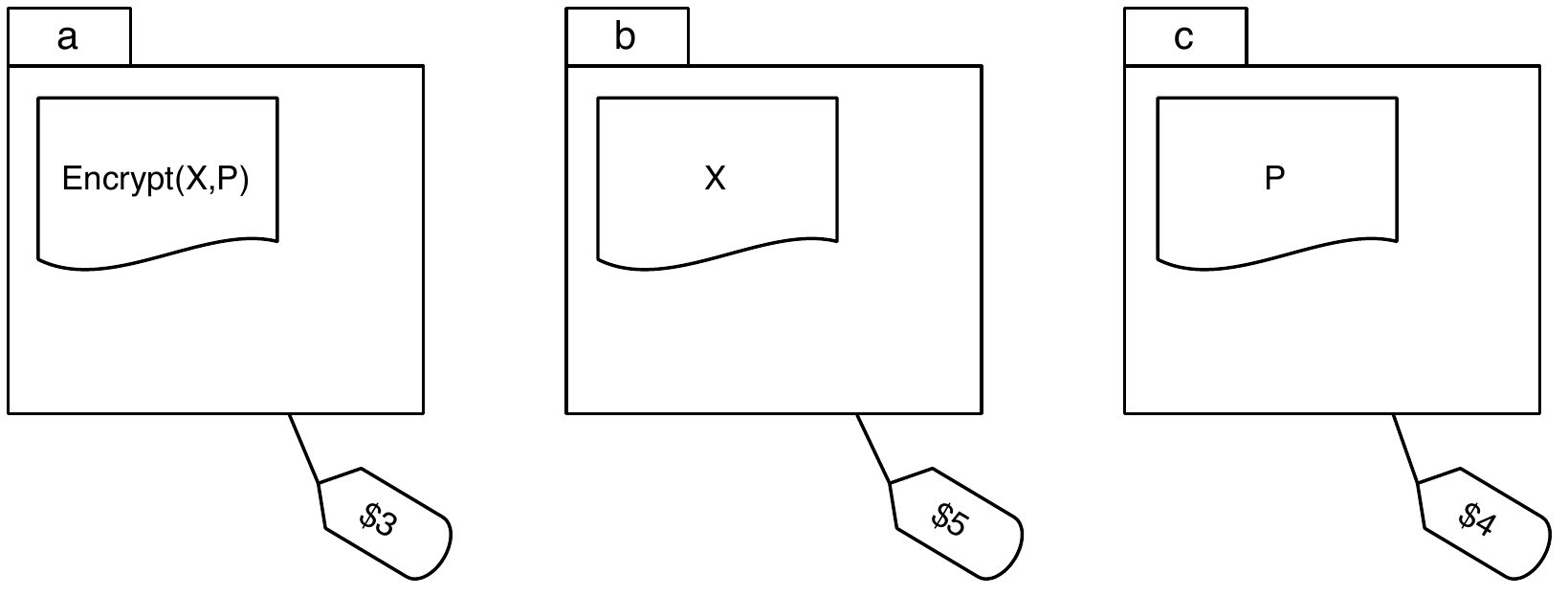}}
\vspace{0mm}
\footnotesize
\caption{Formula $a\rhd_4 b\to \varnothing\rhd_4 b$ is false.}\label{Folders-5 figure}
\vspace{0cm}
\end{center}
\vspace{0mm}
\end{figure}

The one-time pad encryption is known in cryptography as a symmetric-key algorithm because the same key (i.e. the one-time pad) could be used to encrypt and to decrypt the text. As a result, in the model depicted in Figure~\ref{Folders-5 figure}, not only formula $a\rhd_4 b$ is true, but formula $b\rhd_4 a$ is true as well. 

For the next example, we construct a counterexample for formula 
$$a\rhd_4 b\to (\varnothing\rhd_4 b \vee b\rhd_4 a).$$
This is an easier task than one might think because one just needs to modify the previous model by adding add to the folder $a$ some extra document not related to the document $X$ and to raise the price of this folder, see Figure~\ref{Folders-6 figure}. This guarantees that the only way to learn all the content of folder $a$ is to buy folder $a$ directly.

\begin{figure}[ht]
\begin{center}
\vspace{0mm}
\scalebox{.5}{\includegraphics{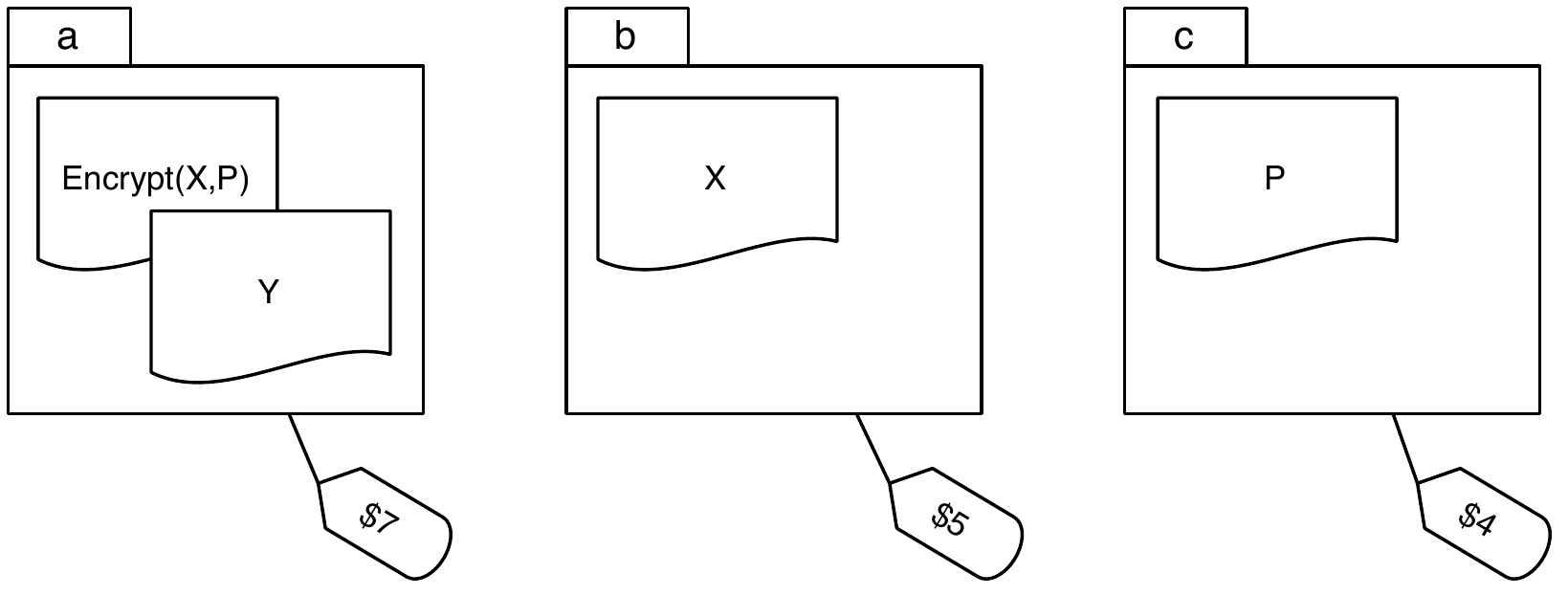}}
\vspace{0mm}
\footnotesize
\caption{Formula $a\rhd_4 b\to (\varnothing\rhd_4 b \vee b\rhd_4 a)$ is false.}\label{Folders-6 figure}
\vspace{0cm}
\end{center}
\vspace{0mm}
\end{figure}

The situation becomes much more complicated if we want (i) the value of attribute $a$ to be recoverable from the value of attribute $b$ and (ii) the value of attribute $b$ to be recoverable from the value of attribute $a$, but at a different price. In other words, we want to construct a counterexample of the following formula:

\begin{equation}\label{big example}
a\rhd_1 b \wedge b\rhd_5 a\to (\varnothing\rhd_5 a \vee \varnothing\rhd_1 b \vee b\rhd_4 a).
\end{equation}
At first glance, this goal could be achieved using the asymmetric key cryptography, commonly used in the public-key encryption. For instance, suppose that folder $a$ contains a document $X$ and folder $b$ contains the same document encrypted with an encryption key $k_e$, see Figure~\ref{Folders-8 figure}. To obtain the content of folder $b$ based on the content of folder $a$, one only needs to know the encryption key $k_e$. To restore the content of folder $a$ based on folder $b$ one needs to know the value of the decryption\footnote{In public-key cryptography, an encryption key is known as the public key and a decryption key as the private key. We do not use these terms here because in our setting neither of the keys is public in the sense that both of them have associated costs.} key $k_d$. If the encryption key and the decryption key are priced at \$1 and \$5 respectively, the formula $b\rhd_4 a$ is not satisfied from the cryptographic point of view. Since folders $a$ and $b$ are priced in this model at \$100 each, formulas $\varnothing\rhd_5 a$ and $\varnothing\rhd_1 b$ are not satisfied either. Thus, the entire formula (\ref{big example}) is not satisfied from the cryptographic point of view.

\begin{figure}[ht]
\begin{center}
\vspace{0mm}
\scalebox{.5}{\includegraphics{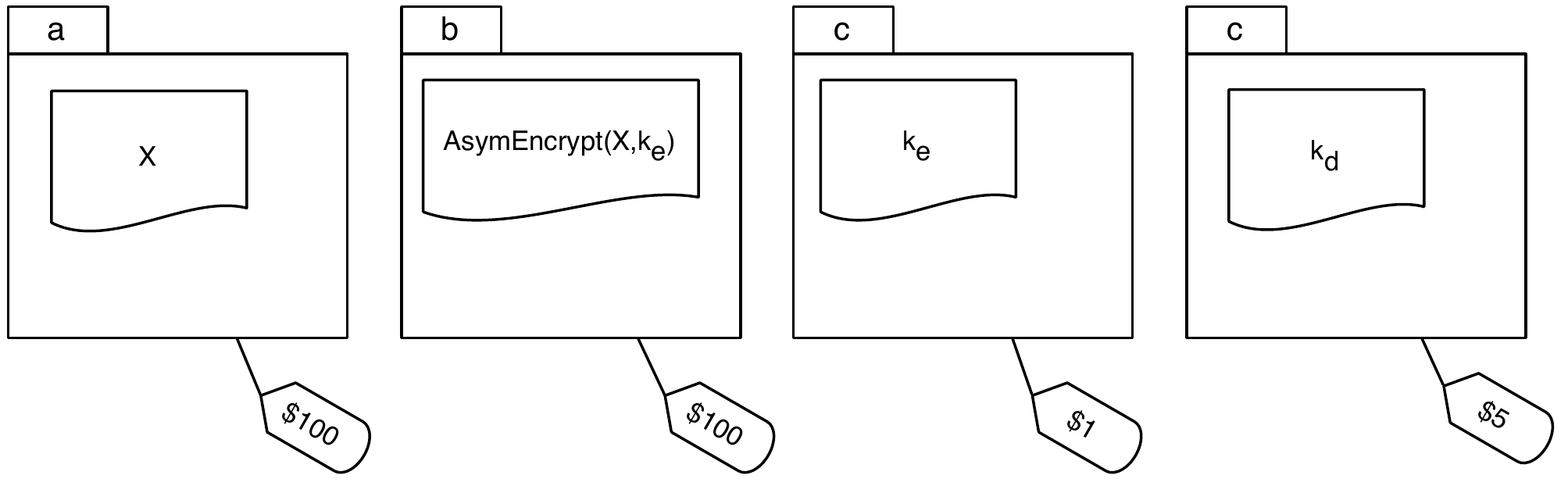}}
\vspace{1mm}
\footnotesize
\caption{Formula $a\rhd_1 b \wedge b\rhd_5 a\to (\varnothing\rhd_5 a \vee \varnothing\rhd_1 b \vee b\rhd_4 a)$ is false.}\label{Folders-8 figure}
\vspace{0cm}
\end{center}
\vspace{0mm}
\end{figure}

Note, however, that cryptographic asymmetric-key algorithms are only polynomial time secure and the proof of polynomial time security requires an appropriate computational hardness  assumption~\cite[Ch. 2]{k10}. In other words, in public-key cryptography, the encrypted text can be decrypted using only the public encryption key if one has exponential time for the decryption. Neither~\cite{a74} definition of functional dependency  nor our definition of budget-constrained functional dependency, given in Definition~\ref{info sat} below, assumes any upper bound on computability of the functional dependency. From our point of view, one would be able to eventually restore the content of folder $a$ based on folder $b$ by spending \$1 on the content of folder $c$. Thus, in the above setting, without polynomial restriction on computability, not only formula $b\rhd_4 a$ is true, but formula $b\rhd_1 a$ is true as well.

\begin{figure}[ht]
\begin{center}
\vspace{0mm}
\scalebox{.5}{\includegraphics{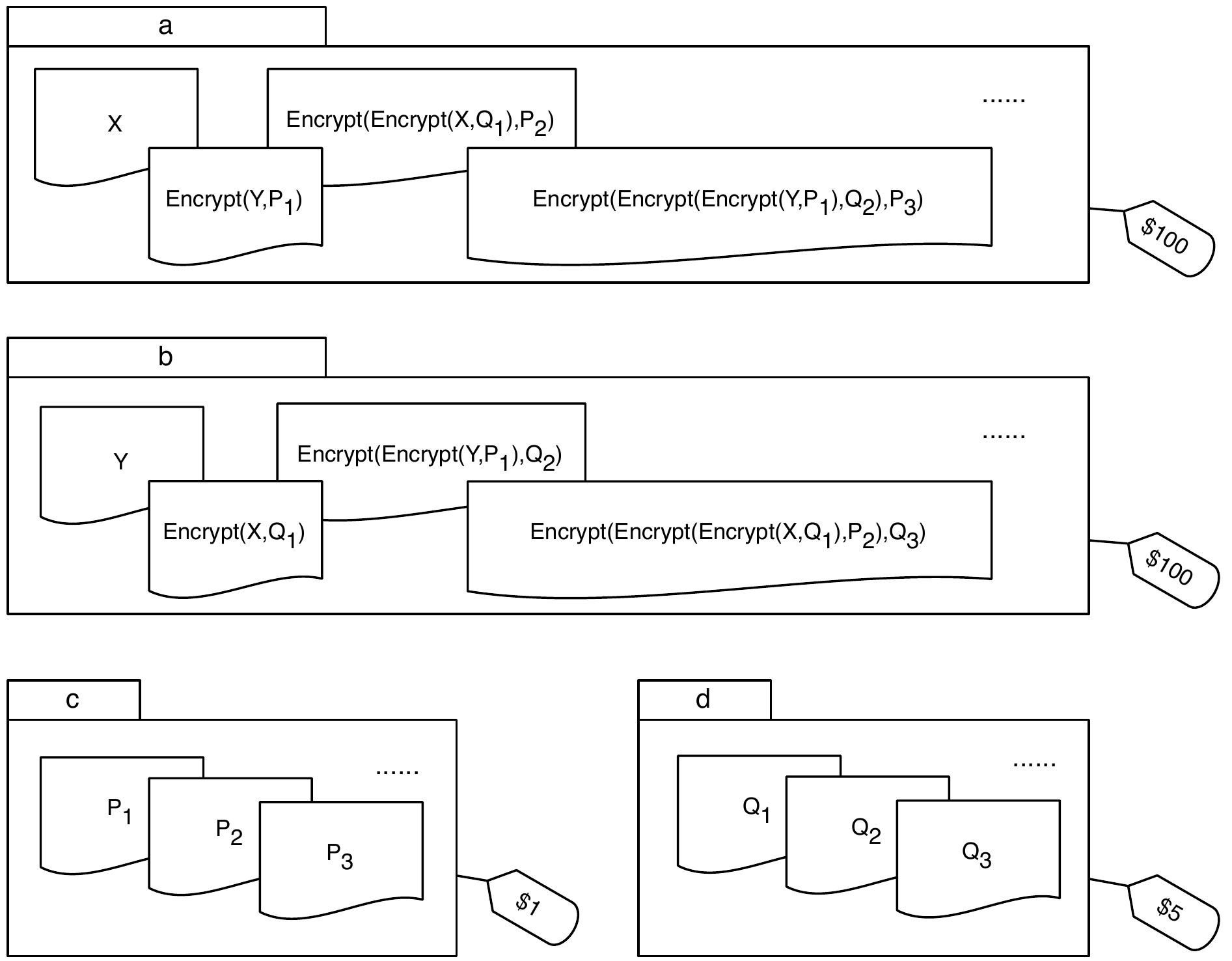}}
\vspace{1mm}
\footnotesize
\caption{Formula $a\rhd_1 b \wedge b\rhd_5 a\to (\varnothing\rhd_5 a \vee \varnothing\rhd_1 b \vee b\rhd_4 a)$ is false.}\label{Folders-7 figure}
\vspace{0cm}
\end{center}
\vspace{0mm}
\end{figure}

Figure~\ref{Folders-7 figure} shows a counterexample for statement (\ref{big example}) that does not require functional dependency to be polynomial time computable. Assume that folders $a$ and $b$ contain copies of unrelated documents $X$ and $Y$, folder $c$ contains an infinite supply of one-time encryption pads $P_1,P_2,P_3,\dots$ and folder $d$ contains another infinite set of one-time encryption pads $Q_1,Q_2,Q_3,\dots$. First, encrypt document $Y$ with one-time pad $P_1$ and place a copy of the resulting cyphertext $Encrypt(Y,P_1)$ into folder $a$. Next, encrypt $Encrypt(Y,P_1)$ with pad $Q_2$ and place a copy of the resulting cyphertext $Encrypt(Encrypt(Y,P_1),Q_2)$ into folder $b$. Then, encrypt $Encrypt(Encrypt(Y,P_1),Q_2)$ with pad $P_3$ and place a copy of the resulting cyphertext $Encrypt(Encrypt(Encrypt(Y,P_1),Q_2),P_3)$ into folder $a$, and so on ad infinitum. Perform similar steps with the document $X$, as shown in Figure~\ref{Folders-7 figure}.

To show that the model depicted in Figure~\ref{Folders-7 figure} is a counterexample for formula (\ref{big example}), we need to prove that both formulas $a\rhd_1 b$ and $b\rhd_5 a$ are satisfied in this model and each of the formulas $\varnothing\rhd_5 a$, $\varnothing\rhd_1 b$,  and $b\rhd_4 a$ is not satisfied. First, notice that formula $a\rhd_1 b$ is satisfied because folder $a$ contains all documents in folder $b$ encrypted with one-time pads $P_1,P_2,\dots$ and that all these pads could be acquired for \$1 by buying folder $c$. Second, formula $b\rhd_5 a$ is satisfied for a similar reason using pads $Q_1,Q_2,\dots$. Third, formula $\varnothing\rhd_5 a$ is not satisfied because for \$5 one can only buy either folder $c$ or folder $d$, both containing only one-time pads. In the absence of folder $b$, one-time encryption pads can not be used to recover document $X$ stored in folder $a$. Formula $\varnothing\rhd_1 b$ is not satisfied for a similar reason. Finally, $b\rhd_4 a$ is not satisfied because \$4 is not enough to buy the content of folder $d$.  This amount of money can only be used to buy pads $P_1,P_2,\dots$ in folder $c$. Knowing the content of folder $b$ and one-time pads $P_1,P_2,\dots$, one can not recover document $X$ contained in folder $a$.

In this article we prove the completeness of our logical system. At the core of this proof is a generalized version of the construction presented in Figure~\ref{Folders-7 figure}.

The axiomatic system proposed in this article is related to other logic system for reasoning about bounded resources.
The classical logical system for reasoning about resources is the linear logic of ~\cite{g87tcs}. 
\cite{al02aamas} presented a family of logical systems for reasoning about beliefs of a perfect reasoner that only can derive consequences of her beliefs after some time delay. This approach has been further developed into the multi-agent Timed Reasoning Logic in~\cite{alw04aamas}. 
\cite{bf10clima} proposed Resource-Bounded Tree Logics for reasoning about resource-bounded computations and obtained preliminary results on the complexity and decidability of model checking for these logics.
\cite{alnr11jlc} incorporated resource requirements into Coalition logic and gave a sound and complete axiomatization of the resulting system. 
Another logical system for reasoning about knowledge under bounded resources was proposed by~\cite{jt13clima}. Their article focuses on the expressive power of the language of the system and the model checking algorithm. 
Our previous work \cite{nt15aamas} introduced a sound and complete modal logic for reasoning about budget-constrained knowledge.
Unlike our current system all of the above logics do not provide a language for expressing functional dependencies.

This paper is organized as follows. In \Section~\ref{syntax and semantics section} we formally define the language of our logical system and its informational semantics. In \Section~\ref{axioms section} we list the axioms of the system that have already been discussed in the introduction. In \Section~\ref{examples section} we give several examples of formal proofs in our logical system. In \Section~\ref{soundness section} we prove the soundness of our axioms with respect to the informational semantics. In \Section~\ref{hypergraph section} we introduce an auxiliary hypergraph semantics of our logical system and prove the completeness with respect to this semantics. In \Section~\ref{info semantics completeness} we use the result obtained in the previous section to prove the completeness of our logical system with respect to the informational semantics. \Section~\ref{finite info semantics completeness} strengthens the informational completeness results by proving the completeness with respect to a more narrow class of {\em finite} informational models. \Section~\ref{conclusion section} concludes the article.

\section{Syntax and Semantics}\label{syntax and semantics section}

In this section we introduce the language of our system and formally describe its intended semantics that we call {\em informational semantics}. Later on we introduce an auxiliary {\em hypergraph semantics} used as a technical tool in the proof of the completeness with respect to the original informational semantics.

\begin{definition}\label{Phi}
For any set $\mathcal{A}$, let language $\Phi(\mathcal{A})$ be the minimum set of formulas such that
\begin{enumerate}
%\item $\bot\in\Phi(\mathcal{A})$,
\item $A\rhd_p B\in \Phi(\mathcal{A})$ for all finite sets $A,B\subseteq \mathcal{A}$ and all non-negative real numbers $p$,
\item if $\phi\in\Phi(\mathcal{A})$, then $\neg\phi\in\Phi(\mathcal{A})$,
\item if $\phi,\psi\in\Phi(\mathcal{A})$, then $\phi\to\psi\in\Phi(\mathcal{A})$.
\end{enumerate}
\end{definition}

Next, we introduce the formal informational semantics of our logical system. The only significant difference between our definition and the one used in~\cite{a74} is the costs function $\|\cdot\|$ that assigns a non-negative cost to each attribute. Note that we assume that the cost is assigned to an attribute, not to its value. For example, if we assign a certain cost to a folder with documents, then this cost is uniform and does not depend on the content of the documents in this folder.

\begin{definition}\label{info model}
An informational model is a tuple $\<\mathcal{A},\{D_a\}_{a\in \mathcal{A}},\|\cdot\|,\mathcal{L}\>$, where
\begin{enumerate}
\item $\mathcal{A}$ is an arbitrary set of ``attributes",
\item $D_a$ is a set representing the domain of attribute $a\in \mathcal{A}$,
\item $\|\cdot\|$ is a cost function that maps each attribute $a\in \mathcal{A}$ into a non-negative real number or infinity $+\infty$,
\item $\mathcal{L}\subseteq \prod_{a\in\mathcal{A}}D_a$ is the set of ``legitimate" vectors of attribute values under the constraints imposed by the informational model.
\end{enumerate}
\end{definition}
In the example depicted in Figure~\ref{Folders-7 figure}, folders are attributes and the information stored in the documents contained in a folder is a value of such an attribute. The set of all possible values of an attribute is its domain. The cost of different attributes is specified explicitly in Figure~\ref{Folders-7 figure}. Note that there is a certain dependency between the plaintext, one-time encryption pads, and the cyphertext. In other words, not all combinations of values of different attributes are possible. The set $\mathcal{L}$ is the set of possible, or ``legitimate", combinations of these values.

We allow the cost $\|a\|$ of an attribute $a$ to be infinity. Informally, one can interpret this as attribute $a$ 
not being available for purchase at any cost. If all attributes are available for sale, then we say that the informational model is {\em finite}. 
\begin{definition}\label{finite model}
Informational model $\<\mathcal{A},\{D_a\}_{a\in \mathcal{A}},\|\cdot\|,\mathcal{L}\>$ is finite if $\|a\|<+\infty$ for each $a\in\mathcal{A}$.
\end{definition}
%
%In this article we first show the completeness of our logical system with respect to the class of all informational models. Then, we strengthen this result buy proving completeness with respect to the class of finite informational models.

\begin{definition}\label{=set}
For any vector $\ell_1=\<f^1_a\>_{a\in\mathcal{A}}\in\mathcal{L}$, any vector $\ell_2=\<f^2_a\>_{a\in\mathcal{A}}\in \mathcal{L}$, and any set $A\subseteq \mathcal{A}$, let $\ell_1=_A\ell_2$ if $f^1_a= f^2_a$ for each attribute $a\in A$.
\end{definition}

\begin{definition}\label{||set}
For each finite set $A\subseteq \mathcal{A}$, let $\|A\|=\sum_{a\in A}\|a\|$.
\end{definition}

The next definition is the key definition of this section. It specifies formal semantics of our logical system. Item 1. of this definition provides the exact meaning of the budget-constrained functional dependency.

\begin{definition}\label{info sat}
For each informational model $I=\<\mathcal{A},\{D_a\}_{a\in \mathcal{A}},\|\cdot\|,\mathcal{L}\>$ and each formula  $\phi\in\Phi(\mathcal{A})$, the satisfiability relation $I\vDash\phi$ is defined as follows:
\begin{enumerate}
\item $I\vDash A\rhd_p B$ when there is a finite set $C\subseteq \mathcal{A}$ such that $\|C\|\le p$ and for each vectors $\ell_1,\ell_2\in\mathcal{L}$, if $\ell_1=_{A,C}\ell_2$, then $\ell_1=_B \ell_2$,
\item $I\vDash \neg\psi$ if $I\nvDash \psi$,
\item $I\vDash \psi\to\chi$ if  $I\nvDash \psi$ or $I\vDash \psi$.
\end{enumerate}
\end{definition}

\section{Axioms}\label{axioms section}

For any set of attributes $\mathcal{A}$, our logical system, in addition to propositional tautologies in language $\Phi(\mathcal{A})$ and the Modus Ponens inference rule, contains the following axioms:

\begin{enumerate}
\item Reflexivity: $A\rhd_p B$, where $B\subseteq A$,
\item Augmentation: $A\rhd_p B\to A,C\rhd_p B,C$,
\item Transitivity: $A\rhd_p B\to(B\rhd_q C\to A\rhd_{p+q} C)$.
\end{enumerate}
By $A,B$ we denote the union of sets $A$ and $B$.
We write $\vdash\phi$ if formula $\phi$ is derivable in our system. Also, we write $X \vdash\phi$ if formula $\phi$ is derivable in our system extended by the set of additional axioms $X$.

\section{Examples of Proofs}\label{examples section}

We prove the soundness of our logical system in the next section. Here we give several examples of formal proofs in this system. We start by showing that the Weakening and the Monotonicity axioms from \cite{v14arxiv} are derivable in our system.

\begin{proposition}[Weakening]
$\vdash A\rhd_p C,D\to A,B\rhd_p C$.
\end{proposition}
\begin{proof}
By the Augmentation axiom,
\begin{equation}\label{eq1}
\vdash A\rhd_p C,D\to A,B\rhd_p B,C,D.
\end{equation}
By the Reflexivity axiom,
\begin{equation}\label{eq2}
\vdash B,C,D\rhd_0 C.
\end{equation}
By the Transitivity axiom,
\begin{equation}\label{eq3}
A,B\rhd_p B,C,D\to(B,C,D\rhd_0 C\to A,B\rhd_p C).
\end{equation}
Finally, from (\ref{eq1}), (\ref{eq2}), and (\ref{eq3}), by the laws of propositional logic,
$$\vdash A\rhd_p C,D\to A,B\rhd_p C.$$
\end{proof}

\begin{proposition}[Monotonicity]
$\vdash A\rhd_p B\to A\rhd_q B$, where $p\le q$.
\end{proposition}
\begin{proof}
By the Reflexivity axiom,
\begin{equation}\label{eq4}
\vdash B\rhd_{q-p} B.
\end{equation}
By the Transitivity axiom,
\begin{equation}\label{eq5}
\vdash A\rhd_p B\to(B\rhd_{q-p} B\to A\rhd_q B).
\end{equation}
Finally, from (\ref{eq4}) and (\ref{eq5}), by the laws of propositional logic,
$$
\vdash A\rhd_p B\to A\rhd_q B.
$$
\end{proof}

As our last example, we prove a generalized version of the Augmentation axiom.
\begin{proposition}
$\vdash A\rhd_p B\to(C\rhd_q D\to A,C\rhd_{p+q}B,D)$.
\end{proposition}
\begin{proof}
By the Augmentation axiom,
\begin{equation}\label{eq11}
\vdash A\rhd_p B \to A,C\rhd_p B,C
\end{equation}
and
\begin{equation}\label{eq12}
\vdash C\rhd_q D \to B,C\rhd_q B,D.
\end{equation}
At the same time, by the Transitivity axiom,
\begin{equation}\label{eq13}
\vdash A,C\rhd_p B,C\to(B,C\rhd_q B,D\to A,C\rhd_{p+q}B,D).
\end{equation}
Finally, from (\ref{eq11}), (\ref{eq12}), and (\ref{eq13}), by the laws of propositional logic,
$$\vdash A\rhd_p B\to(C\rhd_q D\to A,C\rhd_{p+q}B,D).$$
\end{proof}

\section{Soundness}\label{soundness section}

In this section we prove the soundness of our logical system.

\begin{theorem}\label{soundness}
If $\phi\in\Phi(\mathcal{A})$ and $\vdash\phi$, then $I\vDash\phi$ for each informational model $I=\<\mathcal{A},\{D_a\}_{a\in \mathcal{A}},\|\cdot\|,\mathcal{L}\>$.
\end{theorem}
The soundness of propositional tautologies and the Modus Ponens inference rule follows from Definition~\ref{info sat} in the standard way. Below we prove the soundness of the remaining axioms as separate lemmas.

\begin{lemma}
For each finite sets $A,B\subseteq \mathcal{A}$, if $B\subseteq A$, then $I\vDash A\rhd_p B$.
\end{lemma}
\begin{proof}
Let $C=\varnothing$. Thus, $\|C\|=\|\varnothing\|=0\le p$. Consider any two vectors $\ell_1,\ell_2\in\mathcal{L}$ such that $\ell_1=_{A,C}\ell_2$. It suffices to show that $\ell_1=_B\ell_2$, which is true due to Definition~\ref{=set} and the assumption $B\subseteq A$. 
\end{proof}

\begin{lemma}
For each finite sets $A,B,C\subseteq \mathcal{A}$, if $I\vDash A\rhd_p B$, then $I\vDash A,C\rhd_p B,C$.
\end{lemma}
\begin{proof}
By Definition~\ref{info sat}, assumption $I\vDash A\rhd_p B$ implies that there is a set $D\subseteq\mathcal{A}$ such that (i) $\|D\|\le p$ and (ii) for each $\ell_1,\ell_2\in\mathcal{L}$, if $\ell_1=_{A,D}\ell_2$, then $\ell_1=_B\ell_2$. 

Consider now $\ell_1,\ell_2\in\mathcal{L}$ such that $\ell_1=_{A,C,D}\ell_2$. It suffices to show that $\ell_1=_{B,C}\ell_2$. Note that assumption $\ell_1=_{A,C,D}\ell_2$ implies that $\ell_1=_{A,D}\ell_2$ and $\ell_1=_{C}\ell_2$ by Definition~\ref{=set}\idot\ Due to condition (ii) above, the former implies that $\ell_1=_{B}\ell_2$. Finally, statements $\ell_1=_{B}\ell_2$ and $\ell_1=_{C}\ell_2$ together imply that $\ell_1=_{B,C}\ell_2$.
\end{proof}

\begin{lemma}
For each finite sets $A,B,C\subseteq \mathcal{A}$, if $I\vDash A\rhd_p B$ and $I\vDash B\rhd_q C$, then $I\vDash A\rhd_{p+q} C$.
\end{lemma}
\begin{proof}
By Definition~\ref{info sat}, assumption $I\vDash A\rhd_p B$ implies that there is $D_1\subseteq \mathcal{A}$ such that (i) $\|D_1\|\le p$ and (ii) for each $\ell_1,\ell_2\in\mathcal{L}$, if $\ell_1=_{A,D_1}\ell_2$, then $\ell_1=_B\ell_2$.

Similarly, assumption $I\vDash B\rhd_q C$ implies that there is $D_2\subseteq \mathcal{A}$ such that (iii) $\|D_2\|\le q$ and (iv) for each $\ell_1,\ell_2\in\mathcal{L}$, if $\ell_1=_{B,D_2}\ell_2$, then $\ell_1=_C\ell_2$.

Let $D=D_1,D_2$. By Definition~\ref{||set},  $\|D\|\le \|D_1\|+\|D_2\|$. Taking into account statements (i) and (iii) above, we conclude that $\|D\|\le p + q$. Consider any two vectors $\ell_1,\ell_2\in\mathcal{L}$ such that $\ell_1=_{A,D}\ell_2$. It suffices to show that $\ell_1=_C\ell_2$. Indeed, by Definition~\ref{=set}, assumption $\ell_1=_{A,D}\ell_2$ implies that $\ell_1=_{A,D_1}\ell_2$. Hence, $\ell_1=_{B}\ell_2$ due to condition (ii). At the same time, assumption $\ell_1=_{A,D}\ell_2$ also implies that $\ell_1=_{D_2}\ell_2$ by Definition~\ref{=set}\idot Thus, $\ell_1=_{B,D_2}\ell_2$ by Definition~\ref{=set}\idot Therefore, $\ell_1=_{C}\ell_2$ due to condition (iv).
\end{proof}

This concludes the proof of Theorem~\ref{soundness}.

\section{Auxiliary Hypergraph Semantics}\label{hypergraph section}

The main goal of the rest of the article is to prove the completeness of our logical system with respect to the informational semantics. To achieve this goal we introduce the {\em hypergraph} semantics of our logical system and prove that the following statements are equivalent for each $\phi\in \Phi(\mathcal{A})$:
\begin{enumerate}
\item $\phi$ is provable in our logical system,
\item $\phi$ is satisfied in any informational model with a set of attributes $\mathcal{A}$,
%\item $\phi$ is satisfied in any finite informational model with a set of attributes $\mathcal{A}$,
\item $\phi$ is satisfied in any hypergraph with a set of vertices $\mathcal{A}$.
\end{enumerate}
We prove the equivalence of these statements by showing that the first statement implies the second, the second implies the third, and the third implies the first. Note that we have already proved in Theorem~\ref{soundness} that the first statement implies the second one. In the rest of the article we prove that the second statement implies the third and that the third implies the first. Together, these results will imply the soundness and the completeness of our logical system with respect to the informational semantics.

\subsection{Hypergraph Terminology}

Before defining the hypergraph semantics of our logical system, we introduce the basic hypergraph terminology used throughout the rest of the article. In mathematics, a hypergraph is a generalization of a graph in which edges have arbitrary numbers of ends, see \cite{b89}. Our hypergraph semantics is based on weighted directed hypergraphs. In such hypergraphs, edges are directed in the sense that they have multiple tails and multiple heads. For any given edge $e$, we denote these sets by $in(e)$ and $out(e)$. The edges are weighted in the sense that there is a non-negative value assigned to each edge. 

\begin{definition}\label{hyper model}
A weighted directed hypergraph, or just a ``hypergraph", is a tuple $\<V,E,in,out,w\>$, where
\begin{enumerate}
\item $V$ is an arbitrary finite set of ``vertices",
\item $E$ is an arbitrary (possibly infinite) set of ``edges", disjoint with set $V$,
\item $in$ is a function that maps each edge $e\in E$ into set $in(e)\subseteq V$,
\item $out$ is a function that maps each edge $e\in E$ into set $out(e)\subseteq V$,
\item $w$ is a function that maps each edge $e\in E$ into a non-negative real number $w(e)$.
\end{enumerate}
\end{definition}

\begin{definition}
For any hypergraph $\<V,E,in,out,w\>$ and any set $F\subseteq E$, let $w(F)=\sum_{e\in F}w(e)$.
\end{definition}

Next, we define the closure $A^*_F$ of a set of vertices $A$ with respect to a set of edges $F$ of a hypergraph. Informally, the closure $A^*_F$ is a set of all vertices that are reachable from a vertex in set $A$ following the directed edges in set $F$. In order to follow a directed edge $e\in F$, one needs to be able to first reach all vertices in set $in(e)$. To define the closure $A^*_F$, we first define the partial closure $A^k_F$ of all vertices reachable from $A$ through $F$ in no more than $k$ steps:

\begin{definition}\label{AkF}
For each weighted directed hypergraph  $\<V,E,in,out,w\>$, each set $A\subseteq V$, each set $F\subseteq E$, and each non-negative integer $k$, let set $A^k_F$ be defined recursively as follows:
\begin{enumerate}
\item $A^0_F=A$,
\item for any $k\ge 0$,
$$A^{k+1}_F=A^k_F \cup \bigcup_{\{f\in F\;|\;in(f)\subseteq A^k_F\}} out(f).$$
\end{enumerate}
\end{definition}

Figure~\ref{AkF figure} depicts a hypergraph where vertices are represented by circles and edges by ovals. We use arrows to indicate tails and heads of an edge. An arrow from a vertex to an edge indicates that the vertex is a tail of the edge, and an arrow from an edge to a vertex indicates that the vertex is a head of the edge. For example, $in(e_1)=\{v_1,v_2\}$ and $out(e_1)=\{v_3,v_4\}$. The same figure shows partial closures  $A^0_F=\{v_1,v_2\}$, $A^1_F=\{v_1,v_2,v_3,v_4\}$, and $A^2_F=\{v_1,v_2,v_3,v_4,v_5,v_6\}$, where $A=\{v_1,v_2\}$ and $F=\{e_1,e_2\}$.

\begin{figure}[ht]
\begin{center}
\vspace{0mm}
\scalebox{.5}{\includegraphics{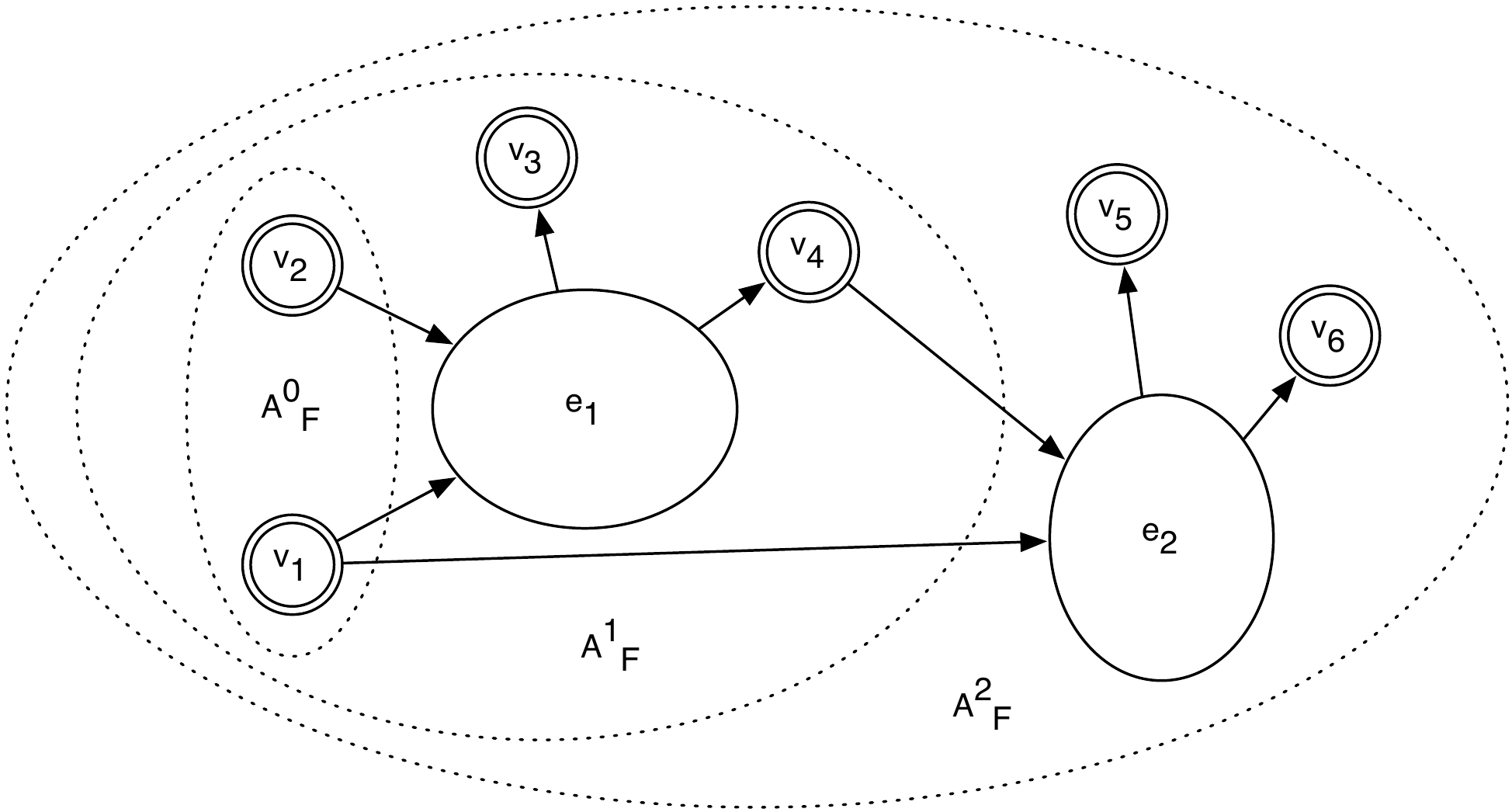}}
\vspace{5mm}
\footnotesize
\caption{$A^k_F$ for $A=\{v_1,v_2\}$ and $F=\{e_1,e_2\}$.}\label{AkF figure}
\vspace{0cm}
\end{center}
\vspace{0mm}
\end{figure}
Finally, we define closure $A^*_F$ to be the union of all partial closures:
\begin{definition}\label{A*F}
$A^*_F=\bigcup_{k\ge 0} A^k_F$.
\end{definition}
Next, we establish two properties of closures that are used later in the proof of the completeness for the hypergraph semantics. 

\begin{lemma}\label{A*F=AkF}
For each set of vertices $A\subseteq V$ and each set of edges $F\subseteq E$ there is $k\ge 0$ such that $A^*_F=A^k_F$. 
\end{lemma}
\begin{proof}
By Definition~\ref{hyper model}, the set of all vertices $V$ is finite. Thus, by Definition~\ref{AkF}, chain
$
A^0_F\subseteq A^1_F\subseteq A^2_F\subseteq \dots
$
is a non-decreasing chain of subsets of finite set $V$. Hence, there must exist $k\ge 0$ such that all sets in this chain starting with set $A^k_F$ are equal. Therefore, $A^*_F=A^k_F$ by Definition~\ref{A*F}\idot
\end{proof}

\begin{lemma}\label{serial lemma}
For each $k\ge 0$, each set of vertices $A\subseteq V$, and each set of edges $F\subseteq E$, there is a sequence 
$A=A_1,f_1,A_2,f_2,\dots,A_{n-1},f_{n-1},A_n=A^k_F$
such that
\begin{enumerate}
\item $n\ge 1$,
\item $f_1,\dots,f_{n-1}\in F$ are distinct edges,
\item $A_1$, $\dots$, $A_{n-1}$, and $A_n$ are subsets of set $V$,
\item $in(f_i)\subseteq A_i$, for each $1\le i< n$,
\item $A_i\cup out(f_i)=A_{i+1}$, for each $1\le i< n$.
\end{enumerate}
\end{lemma}
\begin{proof}
We prove this lemma by induction on $k$. If $k=0$, then $A^k_F=A$ by Definition~\ref{AkF}\idot Therefore, the single-element sequence $A$ is the desired sequence.

For the induction step, assume that there is a sequence 
$$A=A_1,f_1,A_2,f_2,\dots,A_{n-1},f_{n-1},A_n=A^k_F$$ 
that satisfies the conditions 1 through 5 above. Let $g_1,\dots, g_m$ be all such edges $g\in F$ that $in(g)\subseteq A^k_F$ and $out(g)\nsubseteq A^k_F$. By the condition 5 above, the condition $out(g)\nsubseteq A^k_F$ implies that none of $g_1,\dots,g_m$ is equal to any of $f_1,\dots,f_{n-1}$. Note that $A^{k+1}_F=A^{k}_F\cup \bigcup_{i=1}^m out(g_i)$ by Definition~\ref{AkF}\idot Therefore, the two-line sequence
\begin{eqnarray*}
A=A_1,f_1,A_2,f_2,\dots,A_{n-1},f_{n-1},A_n=A^k_F, g_1, A^k_F\cup out(g_1),
\\ g_2, A^k_F\cup out(g_1)\cup out(g_2), g_3, \dots, g_m, A^k_F\cup \bigcup_{i=1}^m out(g_i)=A^{k+1}_F
\end{eqnarray*}
is the required sequence for $k+1$. 
\end{proof}

\subsection{Hypergraph Completeness}

In this section we define the hypergraph semantics of our logical system and prove the completeness of the system with respect to this auxiliary semantics. In other words, using statements defined in the beginning of \Section~\ref{hypergraph section}, we prove that the third statement implies the first one. The hypergraph semantics is specified in the following definition. Item 1 in this definition is the key part because it specifies the meaning of the atomic predicate $A\rhd_p B$. Note that we use symbol $\vDash$ for the satisfiability relation under the informational semantics discussed previously and $\Vdash$ for the satisfiability relation under the hypergraph semantics introduced here.

\begin{definition}\label{hyper sat}
For each hypergraph $H=\<V,E,in,out,w\>$ and each formula  $\phi\in\Phi(V)$, the satisfiability relation $H\Vdash\phi$ is defined as follows:
\begin{enumerate}
\item $H\Vdash A\rhd_p B$ if there is a finite set $F\subseteq E$ such that $w(F)\le p$ and $B\subseteq A^*_{F}$,
\item $H\Vdash \neg\psi$ if  $H\nVdash \psi$,
\item $H\Vdash \psi\to\chi$ if  $H\nVdash \psi$ or $H\Vdash \psi$.
\end{enumerate}
\end{definition}

The next theorem is the completeness theorem for the hypergraph semantics of our logical system. The proof of this theorem ends at the end of this section.
\begin{theorem}\label{hyper completeness}
Let $V$ be a set and $\phi\in\Phi(V)$. If $H\Vdash\phi$ for each hypergraph $H$ with set $V$ as vertices, then $\vdash\phi$.
\end{theorem}
\begin{proof}
Suppose that $\nvdash\phi$. Let $X$ be a maximal consistent subset of $\Phi(V)$ containing formula $\neg\phi$. We define a hypergraph $H=\<V,E,in,out,w\>$ as follows. Let $E$ be set $\{\<A,p,B\>\;|\; A\rhd_p B\in X\}$. For each edge $\<A,p,B\>\in E$, let $in(\<A,p,B\>)=A$, $out(\<A,p,B\>)=B$, and $w(\<A,p,B\>)=p$. In the corollary below and the three lemmas that follow, we establish basic properties of the hypergraph $H$ needed to finish the proof of the completeness.

\begin{corollary}\label{in w out corollary}
$in(e)\rhd_{w(e)}out(e)\in X$ for each $e\in E$.
\end{corollary}

\begin{lemma}\label{*lemma}
$X\vdash A\rhd_p A^*_F$ for each finite $F\subseteq E$ such that $w(F)\le p$ and for each set $A\subseteq V$.
\end{lemma}
\begin{proof}
By  Lemma~\ref{A*F=AkF}, there must exist $k\ge 0$ such that $A^*_F=A^k_F$. Thus, by Lemma~\ref{serial lemma}, there is a sequence 
\begin{equation}\label{star sequence}
A=A_1,f_1,A_2,f_2,\dots,A_{n-1},f_{n-1},A_n=A^*_F
\end{equation}
such that
\begin{enumerate}
\item $n\ge 1$,
\item $f_1,\dots,f_{n-1}\in F$ are distinct edges,
\item $A_1$, \dots, $A_{n-1}$, and $A_n$ are subsets of set $V$,
\item $in(f_i)\subseteq A_i$, for each $1\le i< n$,
\item $A_i\cup out(f_i)=A_{i+1}$, for each $1\le i< n$.
\end{enumerate}
Towards the proof of the lemma, we first show that
\begin{equation}\label{induction claim}
X\vdash A\rhd_{\sum_{i=1}^{m-1} w(f_i)}A_m
\end{equation}
for each $1\le m\le n$. We prove this by induction on $m$. If $m=1$, then $A_m=A$ due to the choice of sequence~(\ref{star sequence}). Thus, $X \vdash A\rhd_0 A_1$ by the Reflexivity axiom. 

Assume that $X\vdash A\rhd_{\sum_{i=1}^{m-1} w(f_i)}A_m$. We need to show that $X\vdash A\rhd_{\sum_{i=1}^{m} w(f_i)}A_{m+1}$. Indeed, since $f_{m}\in F\subseteq E$, then by Corollary~\ref{in w out corollary}, we have $$in(f_{m})\rhd_{w(f_{m})}out(f_{m})\in X.$$ Thus, 
$$
X\vdash A_{m}, in(f_{m})\rhd_{w(f_{m})}A_m, out(f_{m})
$$
by the Augmentation axiom. Note that $in(f_{m})\subseteq A_m$ due to the condition 4 above. Hence,
$$
X\vdash A_{m}\rhd_{w(f_{m})}A_m, out(f_{m}).
$$
Also note that $A_m\cup out(f_{m})=A_{m+1}$ by condition 5 above. Thus,
$$
X\vdash A_{m}\rhd_{w(f_{m})}A_{m+1}.
$$
Therefore, by the induction hypothesis and the Transitivity axiom,
$$X\vdash A\rhd_{\sum_{i=1}^{m} w(f_i)}A_{m+1}.$$

This completes the proof of statement~(\ref{induction claim}). To finish the proof of the lemma, note that the assumption $w(F)\le p$ implies  $\sum_{i=1}^{n-1}w(f_i)\le p$. Thus, $\vdash A\rhd_{p-\sum_{i=1}^{n-1}w(f_i)}A$ by the Reflexivity axiom.  At the same time, statement~(\ref{induction claim}) for $m=n$ asserts that $X\vdash A\rhd_{\sum_{i=1}^{n-1} w(f_i)}A_n$. Hence, by the Transitivity axiom,
$$X\vdash A\rhd_{p-\sum_{i=1}^{n-1}w(f_i)+\sum_{i=1}^{n-1} w(f_i)}A_n.$$
In other words, $X\vdash A\rhd_{p}A_n$. Therefore, $X\vdash A\rhd_p A^*_F$ due to the choice of sequence~(\ref{star sequence}).
\end{proof}

\begin{lemma}\label{hyper iff lemma}
$A\rhd_p B\in X$ if and only if $H\Vdash A\rhd_p B$, for each sets $A, B\subseteq V$ and each non-negative real number $p$.
\end{lemma}
\begin{proof}
\noindent $(\Rightarrow)$. Suppose that $A\rhd_p B\in X$. Then, $\<A,p,B\>\in E$ by the choice of set $E$. Let $F$ be the singleton set $\{\<A,p,B\>\}$. Note that $w(F)=p$, $in(\<A,p,B\>)=A=A^0_F$, and, by Definition~\ref{AkF} and Definition~\ref{A*F},
$$
B=out(\<A,p,B\>)\subseteq \bigcup_{\{f\in F\;|\;in(f)\subseteq A^0_F\}} out(f)\subseteq A^1_F\subseteq A^*_F.
$$
Hence, $B\subseteq A^*_F$. Therefore, $H\Vdash A\rhd_p B$ by Definition~\ref{hyper sat}\idot

\noindent $(\Leftarrow)$. Suppose that $H\Vdash A\rhd_p B$. Then, by Definition~\ref{hyper sat}, there is a finite $F\subseteq E$ such that $w(F)\le p$ and $B\subseteq A^*_F$. Hence, $\vdash A^*_F\rhd_0 B$ by the Reflexivity axiom. Additionally, $X\vdash A\rhd_p A^*_F$ by Lemma~\ref{*lemma}. Thus, $X\vdash A\rhd_p B$ by the Transitivity axiom. Therefore, $A\rhd_p B\in X$ due to the maximality of the set $X$.
\end{proof}

\begin{lemma}\label{hyper induction lemma} $\psi\in X$ if and only if $H\Vdash\psi$ for each $\psi\in\Phi(V)$.
\end{lemma}
\begin{proof}
We prove the lemma by induction on the structural complexity of formula $\psi$. The case when $\psi$ is of form $A\rhd_p B$ follows from Lemma~\ref{hyper iff lemma}. The other cases follow from the maximality and consistency of set $X$ and Definition~\ref{hyper sat} in the standard way.
\end{proof}

To conclude the proof of the theorem, note that assumption $\neg\phi\in X$ implies that $H\nVdash\phi$ by Lemma~\ref{hyper induction lemma} and Definition~\ref{hyper sat}\idot
\end{proof}

%====================================================================================

\section{Completeness theorem for the informational semantics}\label{info semantics completeness}

%====================================================================================

In this section, we prove the completeness theorem for the informational semantics stated in the end of this section as Theorem~\ref{completeness}. This result could be rephrased in terms of statements discussed in the beginning of \Section~\ref{hypergraph section} as: the second statement implies the first one. In the previous section, we have already shown that the third statement implies the first one. Thus, it suffices to prove that the second statement implies the third. In other words, we show that if formula $\phi$ is {\em not} satisfied by a hypergraph $H$, then it is {\em not} satisfied by an informational model $I_H$ constructed from $H$. To prove this, we first describe how to construct an informational model $I_H=\<\mathcal{A},\{D_a\}_{a\in \mathcal{A}},\|\cdot\|,\mathcal{L}\>$ based on a given hypergraph $H=\<V,E,in,out,w\>$. 

Let $\mathcal{A}=V\cup E$. Recall that sets $V$ and $E$ are disjoint due to Definition~\ref{hyper model}\idot
\begin{definition}\label{price info model}
For any attribute $a\in\mathcal{A}=V\cup E$, the cost $\|a\|$ is defined as follows:
$$
\|a\|=
\begin{cases}
w(a), & \mbox{ if $a\in E$},\\
+\infty, & \mbox{ if $a\in V$}.
\end{cases}
$$
\end{definition}

We continue the construction with an auxiliary definition of a path on weighted hypergraph model $H$. It is convenient to distinguish two types of paths: paths that are initiated at a vertex and paths that are initiated at an edge. Note that paths of both types terminate at a vertex.

\begin{definition}
A path initiated at vertex $v_0$ is a finite alternating sequence of vertices and edges $\<v_0,e_1,v_1,e_2,\dots, e_n,v_n\>$ such that
\begin{enumerate}
\item $n\ge 0$,
\item $v_{k-1}\in in(e_k)$ for each $1\le k\le n$,
\item $v_k\in out(e_k)$ for each $1\le k\le n$.
\end{enumerate}
\end{definition}
For example, sequence $\<v_1,e_1,v_4,e_2,v_6\>$ is a path initiated at vertex $v_1$ in the hypergraph depicted in Figure~\ref{AkF figure}.

\begin{definition}
A path initiated at edge $e_1$ is a finite alternating sequence of vertices and edges $\<e_1,v_1,e_2,\dots, e_n,v_n\>$ such that
\begin{enumerate}
\item $n\ge 1$,
\item $v_{k-1}\in in(e_k)$ for each $1< k\le n$,
\item $v_k\in out(e_k)$ for each $1\le k\le n$.
\end{enumerate}
\end{definition}
Sequence $\<e_1,v_4,e_2,v_6\>$ is a path initiated at edge $e_1$ in the hypergraph depicted in Figure~\ref{AkF figure}. 

To understand the rest of the construction of informational model $I_H$, let us consider an analogy between this construction and the informal document/folder model discussed in the introduction. The vertices of the hypergraph could be viewed as folders with multiple documents and directed edges show the process of the dissemination and the encryption of these documents between the folders. In our informational model, for the sake of simplicity, each document consists of just a single bit. For the hypergraph depicted in Figure~\ref{AkF figure}, there might be a document (bit) $x_{v_6}$ initially stored in folder $v_6$. The value of this bit will be disseminated along various paths in the hypergraph and encrypted version of the document will be stored in different vertices (folders) along these paths.
More specifically, the value of each bit stored in a vertex is disseminated {\em against} the direction of each edge leading to this vertex. In our example, because vertex $v_6$ is a head of edge $e_2$ whose tails are $v_1$ and $v_4$,  the encrypted value of $x_{v_6}$ is disseminated between these two tails.
Namely, bit $x_{v_6}$ is represented as a sum of three bits $k_{e_2,v_6}$, $x_{v_1,e_2,v_6}$ and $x_{v_4,e_2,v_6}$ modulo two:
$$
x_{v_6} = k_{e_2,v_6} + x_{v_1,e_2,v_6} + x_{v_4,e_2,v_6} \pmod 2.
$$
One can think of bit $k_{e_2,v_6}$ as an encryption key stored in edge $e_2$ and bits $x_{v_1,e_2,v_6}$ and $x_{v_4,e_2,v_6}$ as encrypted documents distributed between vertices $v_1$ and $v_4$. Since vertex $v_1$ is not a head of any edge in the hypergraph depicted in Figure~\ref{AkF figure}, the value of bit $x_{v_1,e_2,v_6}$ is only stored in vertex $v_1$ and not disseminated any further. At the same time, because vertex $v_4$ is a head of edge $e_1$, the value of bit $x_{v_4,e_2,v_6}$ is further distributed between tails of edge $e_1$. It is represented as a sum of three bits $k_{e_1,v_4,e_2,v_6}$, $x_{v_1,e_1,v_4,e_2,v_6}$ and $x_{v_2,e_1,v_4,e_2,v_6}$ modulo two:
$$
x_{v_4,e_2,v_6} = k_{e_1,v_4,e_2,v_6} + x_{v_1,e_1,v_4,e_2,v_6} + x_{v_2,e_1,v_4,e_2,v_6} \pmod 2.
$$ 
Again, one can think of bit $k_{e_1,v_4,e_2,v_6}$ as an encryption key stored at edge $e_1$ and bits $x_{v_1,e_1,v_4,e_2,v_6}$ and $x_{v_2,e_1,v_4,e_2,v_6}$ as encrypted documents distributed between vertices $v_1$ and $v_2$. 

To summarize, informally, each edge in the hypergraph stores one encryption key corresponding to each path initiated at this edge. Vertices store encrypted documents as they are being disseminated along the paths.
 This intuition is captured  
in the two definitions below.

\begin{definition}\label{Da def}
For any attribute $a\in \mathcal{A}=V\cup E$, let domain $D_a$ be defined as follows:
\begin{enumerate}
\item If $a\in V$, then $D_a$ is the set of all functions that map paths initiated at vertex $a$ into set $\{0,1\}$.
\item If $a\in E$, then $D_a$ is the set of all functions that map paths initiated at edge $a$ into set $\{0,1\}$.
\end{enumerate}
\end{definition}

\begin{definition}\label{canonical L}
Let $\mathcal{L}$ be the set of all vectors $\<f_a\>_{a\in\mathcal{A}}\in \prod_{a\in\mathcal{A}}D_a$ such that for each edge-initiated path $\<e_1,v_1,e_2,\dots, e_n,v_n\>$, the following equation is satisfied:
\begin{eqnarray}
f_{e_1}(\<e_1,v_1,e_2,\dots, e_n,v_n\>)+
\sum_{u\in in(e_1)}f_u(\<u,e_1,v_1,e_2,\dots, e_n,v_n\>)
= \nonumber\\
f_{v_1}(\<v_1,e_2,\dots, e_n,v_n\>) \pmod 2.\label{canonical L eq}
\end{eqnarray}
\end{definition}
This concludes the definition of the informational model $I_H=\<\mathcal{A},\{D_a\}_{a\in \mathcal{A}},\|\cdot\|,\mathcal{L}\>$.

\begin{definition}
Let $\mathbf{0}$ be the vector $\<f_a\>_{a\in\mathcal{A}}$ such that $f_a= 0$ for each $a\in\mathcal{A}$.
\end{definition}

\begin{lemma}
$\mathbf{0}\in\mathcal{L}$.
\end{lemma}
\begin{proof}
Equation~(\ref{canonical L eq}) holds for vector $\mathbf{0}$ because $0 + \sum_{u\in in(e_1)}0 = 0 \pmod 2$.
\end{proof}

The dissemination of encrypted information described above on the hypergraph depicted in Figure~\ref{AkF figure} takes place from a head vertex of an edge to the tail vertices of the edge. This means that the bit stored at the head vertex could be determined based on the encryption key and the bits stored at the tail vertices. In other words, information flows in the direction that is opposite to the direction of the edge, but the functional dependency exists in the same direction as the edge. This observation is the key to the understanding of the next lemma.

\begin{lemma}\label{AkF lemma}
For any set of vertices $A\subseteq V$, any set of edges $F\subseteq E$, any $k\ge 0$, and any two vectors $\ell_1,\ell_2\in\mathcal{L}$, if $\ell_1=_{A,F}\ell_2$, then $\ell_1=_{A^k_F} \ell_2$.
\end{lemma}
\begin{proof}
We prove the lemma by induction on $k$. If $k=0$, then $A^k_F=A$ by Definition~\ref{AkF}\idot Thus, assumption $\ell_1=_{A,F} \ell_2$ implies that $\ell_1=_{A} \ell_2$ by Definition~\ref{=set}\idot

Suppose that $\ell_1=_{A^k_F} \ell_2$. We need to prove that $\ell_1=_{A^{k+1}_F} \ell_2$. By Definition~\ref{AkF}, it suffices to prove that $\ell_1=_b \ell_2$ for each $e\in F$ and each $b\in out(e)$, where $in(e)\subseteq A^k_F$. See Figure~\ref{path proof figure}.

\begin{figure}[ht]
\begin{center}
\vspace{0mm}
\scalebox{.5}{\includegraphics{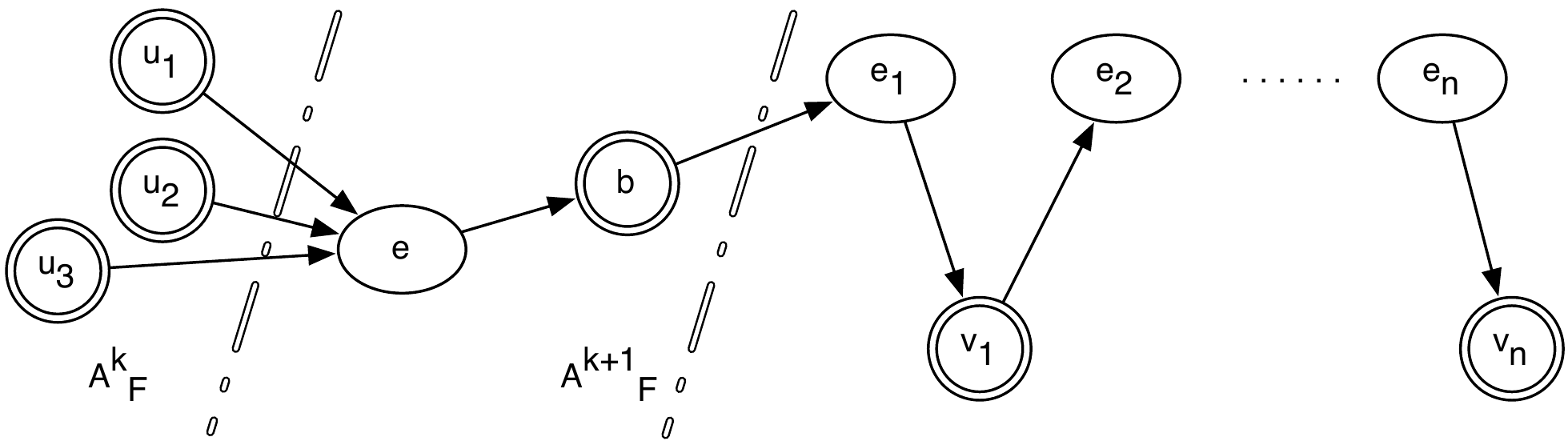}}
\vspace{5mm}
\footnotesize
\caption{Illustration to Lemma~\ref{AkF lemma}.}\label{path proof figure}
\vspace{0cm}
\end{center}
\vspace{0mm}
\end{figure}

Let $\ell_1=\<f^1_a\>_{a\in\mathcal{A}}$ and $\ell_2=\<f^2_a\>_{a\in\mathcal{A}}$. It suffices to show that 
$$f^1_b(\<b,e_1,v_1,\dots,e_n,v_n\>)=f^2_b(\<b,e_1,v_1,\dots,e_n,v_n\>)$$
for each path $\<b,e_1,v_1,\dots,e_n,v_n\>$ initiated at vertex $b$. Indeed, by Definition~\ref{canonical L},
\begin{eqnarray*}
f^1_b(\<b,e_1,v_1,\dots,e_n,v_n\>) &=& f^1_e(\<e,b,e_1,v_1,\dots,e_n,v_n\>) +\\
 && \sum_{u\in in(e)}f^1_u(\<u,e,b,e_1,v_1,\dots,e_n,v_n\>) \pmod 2.
\end{eqnarray*}
Recall that $e\in F$, and so, by Definition~\ref{=set}, assumption $\ell_1=_{A,F}\ell_2$ of the lemma implies  that $\ell_1=_e \ell_2$. Hence, $f^1_e = f^2_e$. Then,
\begin{eqnarray*}
f^1_b(\<b,e_1,v_1,\dots,e_n,v_n\>) &=& f^2_e(\<e,b,e_1,v_1,\dots,e_n,v_n\>) +\\
 && \sum_{u\in in(e)}f^1_u(\<u,e,b,e_1,v_1,\dots,e_n,v_n\>) \pmod 2.
\end{eqnarray*}
By induction hypothesis, $\ell_1=_{A^k_F} \ell_2$. Additionally, $in(e)\subseteq A^k_F$ by the choice of edge $e$. Thus, $\ell_1=_{in(e)} \ell_2$ by Definition~\ref{=set}\idot Hence, $f^1_u = f^2_u$ for each $u\in in(e)$. Then,
\begin{eqnarray*}
f^1_b(\<b,e_1,v_1,\dots,e_n,v_n\>) &=& f^2_e(\<e,b,e_1,v_1,\dots,e_n,v_n\>) +\\
 && \sum_{u\in in(e)}f^2_u(\<u,e,b,e_1,v_1,\dots,e_n,v_n\>) \pmod 2.
\end{eqnarray*}
Therefore, 
$f^1_b(\<b,e_1,v_1,\dots,e_n,v_n\>)=f^2_b(\<b,e_1,v_1,\dots,e_n,v_n\>)$,
by Definition~\ref{canonical L}\idot
\end{proof}

\begin{lemma}\label{A*F lemma}
For any set of vertices $A\subseteq V$, any set of edges $F\subseteq E$, and any two vectors $\ell_1,\ell_2\in\mathcal{L}$, if $\ell_1=_{A,F}\ell_2$, then $\ell_1=_{A^*_F} \ell_2$.
\end{lemma}
\begin{proof}
The statement of the lemma follows from Lemma~\ref{AkF lemma} and Definition~\ref{A*F}\idot
\end{proof}

A cut $(V_1,V_2)$ is a partition of the set of all vertices $E$. We consider cuts to be directed in the sense that  $(V_1,V_2)$ and $(V_2,V_1)$ are two different cuts.

\begin{definition}\label{cross def}
An edge $e\in E$ is called a crossing edge of a cut $(V_1,V_2)$, if $in(e)\subseteq V_1$ and $out(e)\cap V_2\neq \varnothing$. 
\end{definition}

\begin{figure}[ht]
\begin{center}
\vspace{0mm}
\scalebox{.5}{\includegraphics{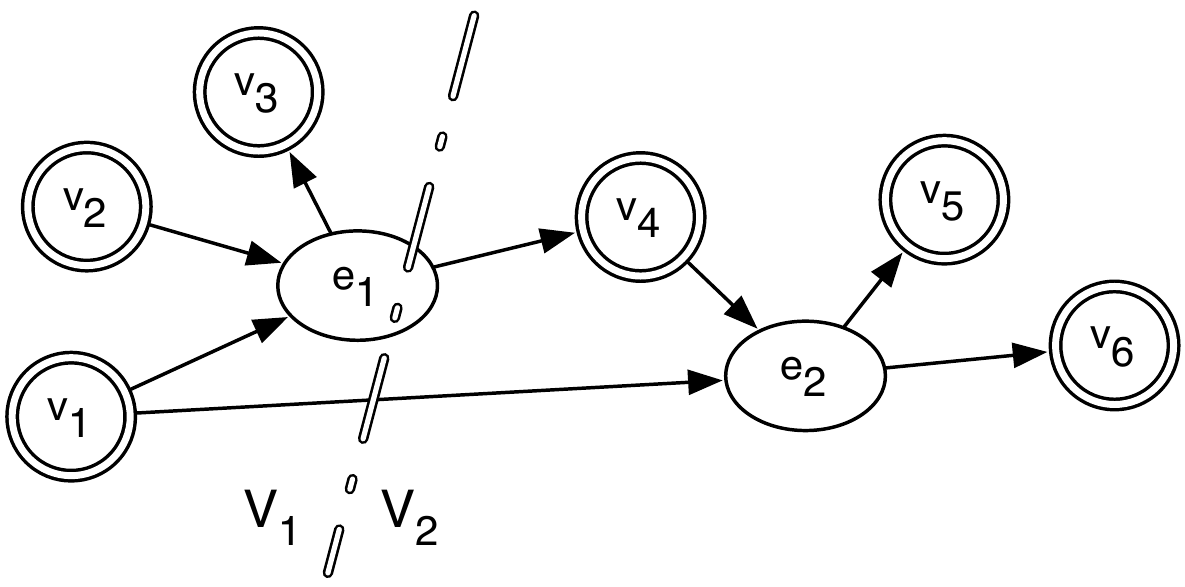}}
\vspace{5mm}
\footnotesize
\caption{$Cross(c)=\{e_1\}$, where $c=(\{v_1,v_2,v_3\},\{v_4,v_5,v_6\})$ }\label{Cross figure}
\vspace{0cm}
\end{center}
\vspace{0mm}
\end{figure}

The set of all crossing edges of cut $c=(V_1,V_2)$ is denoted by $Cross(c)$. For example, edge $e_1$ is the only crossing edge of the cut $c$ depicted in Figure~\ref{Cross figure}. Generally speaking, cuts $(V_1,V_2)$ and $(V_2,V_1)$ have different sets of crossing edges.

\begin{lemma}\label{backtrack lemma}
For each $c=(V_1,V_2)$, if $e\notin Cross(c)$ and $out(e)\cap V_2\neq \varnothing$, then $in(e)\cap V_2\neq\varnothing$. 
\end{lemma}
\begin{proof}
Suppose that $in(e)\cap V_2=\varnothing$. Then, $in(e)\subseteq V_1$. Thus, we have  $out(e)\cap V_2\neq \varnothing$ and $in(e)\subseteq V_1$. Therefore, $e\in Cross(c)$ by Definition~\ref{cross def}\idot
\end{proof}

\begin{lemma}\label{FCross}
If $c=(A^*_F,V\setminus A^*_F)$, then 
$Cross(c)\cap F=\varnothing$, for each set of vertices $A\subseteq V$ and each set of edges $F\subseteq E$.
\end{lemma}
\begin{proof}
Suppose that there is an edge $e\in F$ such that $e\in Cross(c)$. Thus, $in(e)\subseteq A^*_F$ and $out(e)\cap (V\setminus A^*_F)\neq \varnothing$, by Definition~\ref{cross def}\idot By Lemma~\ref{A*F=AkF}, there is $k\ge 0$ such that $A^*_F=A^k_F$. Hence, $in(e)\subseteq A^k_F$. Thus, $out(e)\subseteq A^{k+1}_F$ by Definition~\ref{AkF}\idot Therefore, by Definition~\ref{A*F}, $out(e)\subseteq A^*_F$, which yields a contradiction with $out(e)\cap (V\setminus A^*_F)\neq \varnothing$.
\end{proof}

The next definition specifies the core construction in the proof of the completeness by introducing the notion of a {\em cut-limited inverted tree rooted at a vertex $v$}. Informally, such a tree starts at vertex $v$ and grows through the edges and vertices of the hypergraph.  The tree is {\em inverted} because it grows in the direction against that of the edges. From each vertex, the tree brunches into each edge that has this vertex as a head. From each edge, the tree expands through only one of the tail vertices. Furthermore, if the edge is a crossing edge of the cut, then the tree does not expand from this edge at all. The tree can potentially loop through the hypergraph and be infinite. 
Figure~\ref{Tree figure} shows a cut-limited inverted tree rooted at vertex $m$. Note that this tree is inverted as it expands in the direction opposite to the direction of the edges. The tree is limited by cut $(V_1,V_2)$ and, thus, it does not continue through the crossing edge $r$ of this cut. The tree brunches at vertex $d$ and continues through edges $r$, $s$, and $t$ of which vertex $d$ is a head. Since the tree does not brunch at edges, edge $y$ can only expand  either through tail $h$ or through tail $k$. The inverted tree depicted in Figure~\ref{Tree figure} expands through tail $h$.

\begin{figure}[ht]
\begin{center}
\vspace{0mm}
\scalebox{.5}{\includegraphics{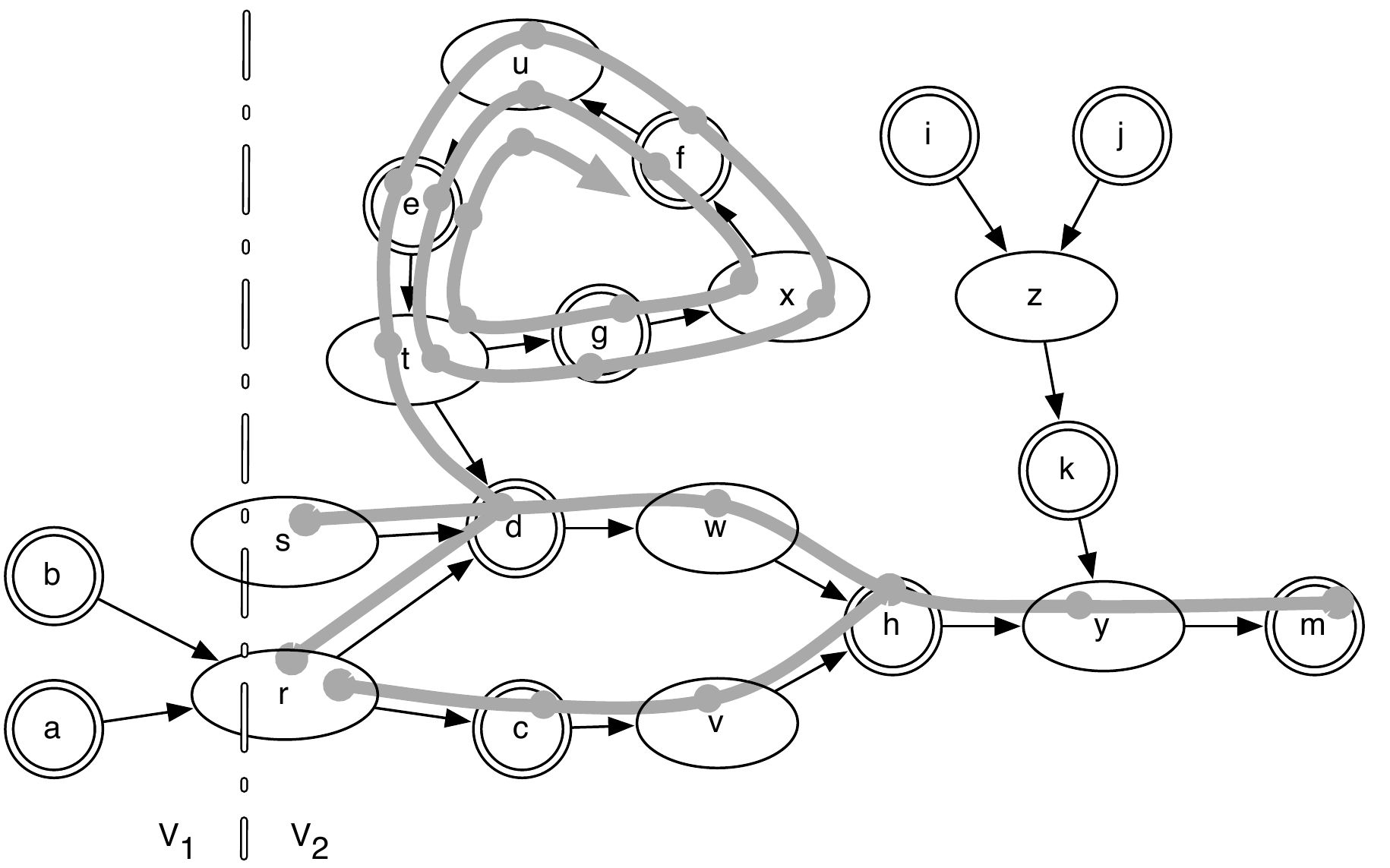}}
\vspace{5mm}
\footnotesize
\caption{A cut-limited inverted tree rooted at vertex $m$.}\label{Tree figure}
\vspace{0cm}
\end{center}
\vspace{0mm}
\end{figure}

The next definition specifies the notion of a cut-limited inverted tree. We formally represent tree as a collection of paths.

\begin{definition}\label{tree def}
For any cut $c=(V_1,V_2)$, $c$-limited inverted tree rooted at vertex $v\in V_2$  is any minimal set of paths $T$ such that
\begin{enumerate}
\item $\<v\>\in T$,
\item for each vertex-initiated path $\<v_0,e_1,v_1,e_1,\dots, e_n,v_n\>\in T$ and edge $e_0\in E$ such that $v_0\in out(e_0)$, we have $\<e_0,v_0,e_1,v_1,e_1,\dots, e_n,v_n\>\in T$,
%\item $\<e_0,v_0,e_1,v_1,e_1,\dots, e_n,v_n\>\in T$ for each edge $e_0\in E$ and each vertex-initiated path $\<v_0,e_1,v_1,e_1,\dots, e_n,v_n\>\in T$ such that $v_0\in out(e_0)$, 
\item for each edge-initiated path $\<e_1,v_1,e_1,\dots, e_n,v_n\>\in T$ if $e_1\notin Cross(c)$ then there is exactly one $v_0\in in(e_1)\cap V_2$ such that $\<v_0,e_1,v_1,e_1,\dots, e_n,v_n\>\in T$.
\end{enumerate}
\end{definition}

Because set $T$ in the above definition is required to be a minimal set satisfying the given conditions, all paths in this set terminate with the vertex $v$ at which the tree is ``rooted".

\begin{lemma}\label{tree exists}
For any cut $c=(V_1,V_2)$ and any $v\in V_2$ there is a $c$-limited inverted tree rooted at vertex $v$.
\end{lemma}
\begin{proof}
We recursively construct an infinite sequence of sets of paths $T_0,T_1,T_2,\dots$, whose vertices are all in set $V_2$, as follows:
\begin{enumerate}
\item $T_0=\{\<v\>\}$.
\item For each $k\ge 0$, let $$T_{2k+1}=\{\<e,v_1,e_1,\dots,v_n,e_n,v\>\;|\; \<v_1,e_1,\dots,v_n,e_n,v\>\in T_{2k},v_1\in out(e)\}.$$
\item For each path $\pi=\<e_0,v_1,e_1,\dots,v_n,e_n,v\>\in T_{2k+1}$ such that $e_0\notin Cross(c)$, choose any vertex $u\in in(e_0)\cap V_2$. Since $v_1\in out(e_0)\cap V_2$, such vertex $u$ exists by Lemma~\ref{backtrack lemma}. Construct a path $\<u,e_0,v_1,e_1,\dots,v_n,e_n,v\>$. Let $T_{2k+2}$ be the set of all paths constructed in such a way, taking only one path (and only one vertex $u$) for each path $\pi$. 
\end{enumerate}
Let $T=\bigcup_{i\ge 0}T_i$.
\end{proof}

\begin{lemma}\label{tree bounded lemma}
For any cut $c=(V_1,V_2)$, any $c$-limited inverted tree rooted at vertex $v\in V_2$, and any $\pi\in T$, all vertices in path $\pi$ belong to set $V_2$.
\end{lemma}
\begin{proof}
The statement of the lemma follows from condition 3 of Definition~\ref{tree def} and the minimality condition on $T$ of the same definition.
\end{proof}

The next lemma shows that two vectors can agree on a large set of attributes while not being equal on all attributes.

\begin{lemma}\label{flip lemma}
For any vector $\<f_a\>_{a\in\mathcal{A}}\in \mathcal{L}$, any cut $c=(V_1,V_2)$, and any $b\in V_2$, there is a vector $\<f'_a\>_{a\in\mathcal{A}}\in \mathcal{L}$ such that
\begin{enumerate}
\item $f'_u = f_u$ for each $u\in V_1$,
\item $f'_e = f_e$ for each $e\in E\setminus Cross(c)$,
\item $f'_b \neq f_b$.
\end{enumerate}
\end{lemma}
\begin{proof}
By Lemma~\ref{tree exists}, there exists a $c$-limited inverted tree $T$ rooted at vertex $b\in V_2$. Define vector $\<f'_a\>_{a\in\mathcal{A}}$ as follows:
\begin{equation}\label{f' def}
f'_a(\pi)=
\begin{cases}
1 + f_a(\pi), & \mbox{if $a\in V$ and $\pi\in T$},\\
1 + f_a(\pi), & \mbox{if $a\in Cross(c)$ and $\pi\in T$},\\
f_a(\pi), & \mbox{otherwise,}
\end{cases}
\pmod 2,
\end{equation}
where $\pi$ is a path initiated at $a$.
\begin{claim}
$\<f'_a\>_{a\in\mathcal{A}}\in\mathcal{L}$.
\end{claim}
\begin{proof}
We need to show that vector $\<f'_a\>_{a\in\mathcal{A}}$ satisfies equation~(\ref{canonical L eq}) of Definition~\ref{canonical L} for each edge-initiated path. There are three cases that we consider separately:

\begin{figure}[ht]
\begin{center}
\vspace{0mm}
\scalebox{.5}{\includegraphics{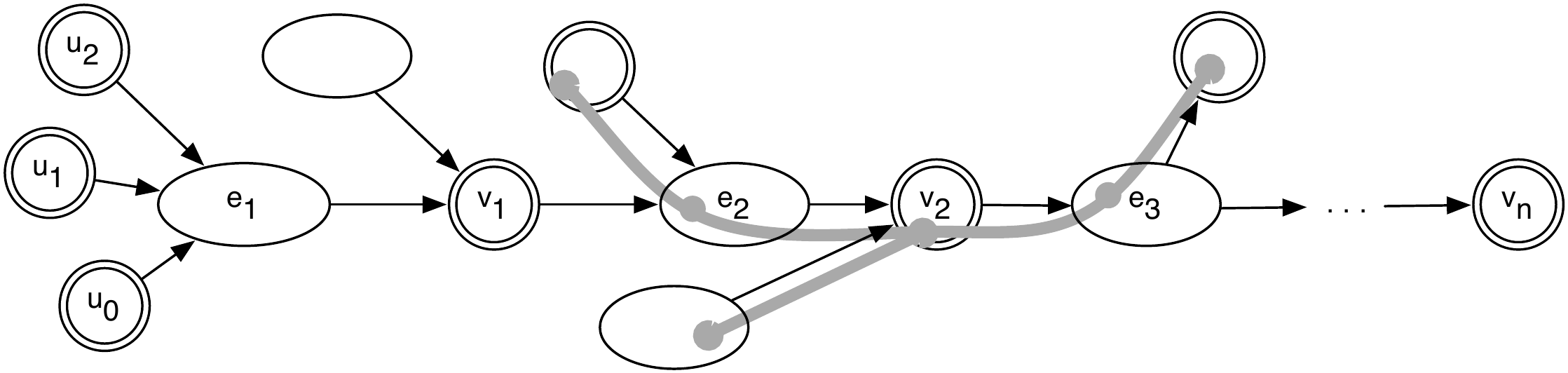}}
\vspace{5mm}
\footnotesize
\caption{Case I}\label{Case I}
\vspace{0cm}
\end{center}
\vspace{0mm}
\end{figure}

\noindent{\em Case I:} path $\<e_1,v_1,e_2,\dots, e_n,v_n\>$ does not belong to tree $T$. In this case, by Definition~\ref{tree def}, path $\<v_1,e_2,\dots, e_n,v_n\>$ does not belong to tree $T$ either. Neither do paths $\<u,e_1,v_1,e_2,\dots, e_n,v_n\>$ for each $u\in in(e_1)$. Thus, according to definition~(\ref{f' def}),
\begin{eqnarray*}
&&f'_{e_1}(\<e_1,v_1,e_2,\dots, e_n,v_n\>) + \sum_{u\in in(e_1)}f'_u(\<u,e_1,v_1,e_2,\dots, e_n,v_n\>)\nonumber\\
&=&f_{e_1}(\<e_1,v_1,e_2,\dots, e_n,v_n\>) + \sum_{u\in in(e_1)}f_u(\<u,e_1,v_1,e_2,\dots, e_n,v_n\>).
\end{eqnarray*}
Since $\<f_a\>_{a\in\mathcal{A}}\in\mathcal{L}$, by Definition~\ref{canonical L},
\begin{eqnarray*}
&&f_{e_1}(\<e_1,v_1,e_2,\dots, e_n,v_n\>)+
\sum_{u\in in(e_1)}f_u(\<u,e_1,v_1,e_2,\dots, e_n,v_n\>)
\nonumber\\
&=&f_{v_1}(\<v_1,e_2,\dots, e_n,v_n\>) \pmod 2.
\end{eqnarray*}
At the same time,  since path $\<v_1,e_2,\dots, e_n,v_n\>$ does not belong to tree $T$, by definition~(\ref{f' def}), we have $f_{v_1}(\<v_1,e_2,\dots, e_n,v_n\>) = f'_{v_1}(\<v_1,e_2,\dots, e_n,v_n\>)$. Therefore,  
\begin{eqnarray*}
&&f'_{e_1}(\<e_1,v_1,e_2,\dots, e_n,v_n\>)+
\sum_{u\in in(e_1)}f'_u(\<u,e_1,v_1,e_2,\dots, e_n,v_n\>) \nonumber\\
&=&f'_{v_1}(\<v_1,e_2,\dots, e_n,v_n\>) \pmod 2.
\end{eqnarray*}

\begin{figure}[ht]
\begin{center}
\vspace{0mm}
\scalebox{.5}{\includegraphics{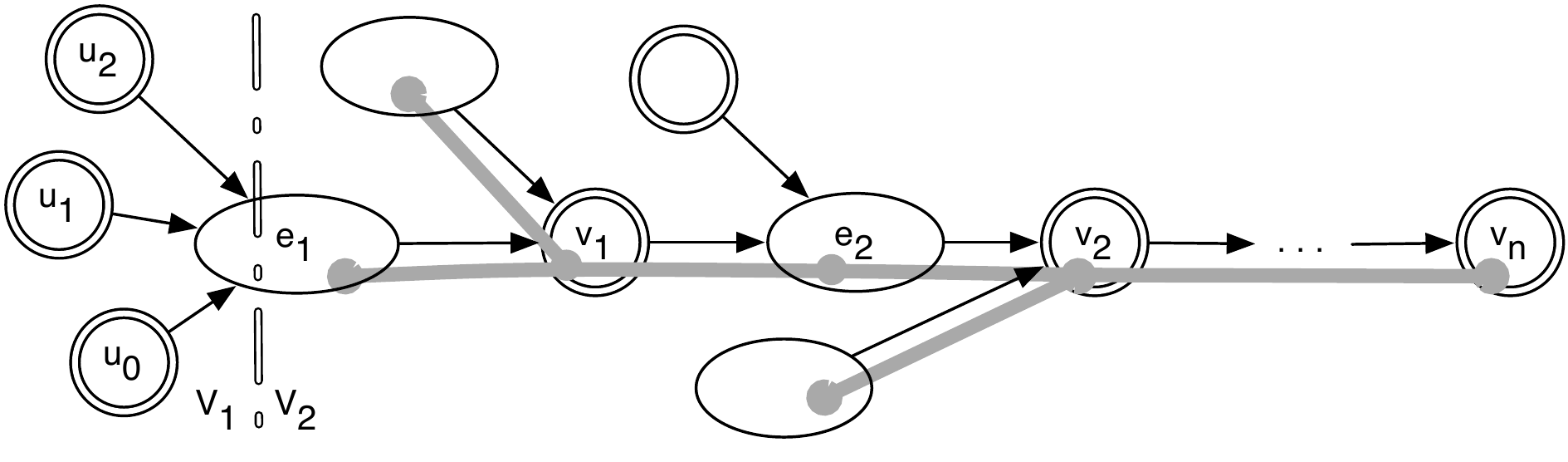}}
\vspace{5mm}
\footnotesize
\caption{Case II}\label{Case II}
\vspace{0cm}
\end{center}
\vspace{0mm}
\end{figure}

\noindent{\em Case II:} path $\<e_1,v_1,e_2,\dots, e_n,v_n\>$ belongs to tree $T$ and $e_1\in Cross(c)$. It follows from Definition~\ref{tree def} that path $\<v_1,e_2,\dots, e_n,v_n\>$ belongs to tree $T$ as well and that path $\<u,e_1,v_1,e_2,\dots, e_n,v_n\>$ does not belong to tree $T$ for each $u\in in(e_1)$. Hence, by definition~(\ref{f' def}),
\begin{eqnarray*}
&&
f'_{e_1}(\<e_1,v_1,e_2,\dots, e_n,v_n\>) + \sum_{u\in in(e_1)}f'_u(\<u,e_1,v_1,e_2,\dots, e_n,v_n\>) \nonumber\\
&=&f_{e_1}(\<e_1,v_1,e_2,\dots, e_n,v_n\>) + 1 + \sum_{u\in in(e_1)}f_u(\<u,e_1,v_1,e_2,\dots, e_n,v_n\>) \pmod 2.
\end{eqnarray*}
Since $\<f_a\>_{a\in\mathcal{A}}\in\mathcal{L}$, by Definition~\ref{canonical L},
\begin{eqnarray*}
&&f_{e_1}(\<e_1,v_1,e_2,\dots, e_n,v_n\>)+
\sum_{u\in in(e_1)}f_u(\<u,e_1,v_1,e_2,\dots, e_n,v_n\>) \nonumber\\
&=& f_{v_1}(\<v_1,e_2,\dots, e_n,v_n\>) \pmod 2.
\end{eqnarray*}
At the same time, since path $\<v_1,e_2,\dots, e_n,v_n\>$ belongs to tree $T$,  by definition~(\ref{f' def}) we have
$f'_{v_1}(\<v_1,e_2,\dots, e_n,v_n\>) = f_{v_1}(\<v_1,e_2,\dots, e_n,v_n\>) + 1 \pmod 2$. Therefore,
\begin{eqnarray*}
&&f'_{e_1}(\<e_1,v_1,e_2,\dots, e_n,v_n\>) + \sum_{u\in in(e_1)}f'_u(\<u,e_1,v_1,e_2,\dots, e_n,v_n\>) \nonumber\\
&=&f_{v_1}(\<v_1,e_2,\dots, e_n,v_n\>) + 1 =
f'_{v_1}(\<v_1,e_2,\dots, e_n,v_n\>) \pmod 2.
\end{eqnarray*}

\begin{figure}[ht]
\begin{center}
\vspace{0mm}
\scalebox{.5}{\includegraphics{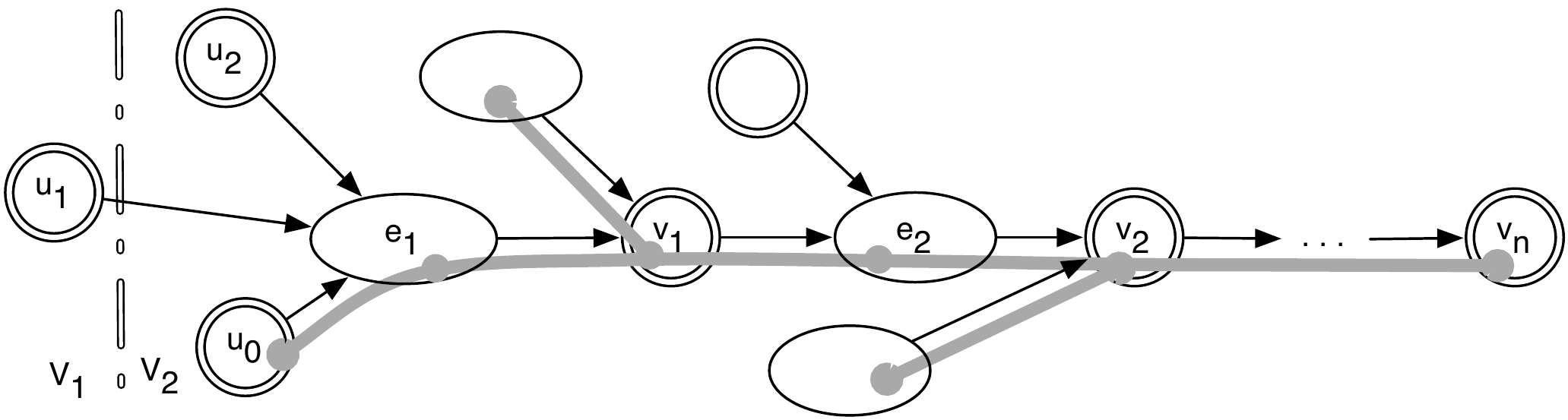}}
\vspace{5mm}
\footnotesize
\caption{Case III}\label{Case III}
\vspace{0cm}
\end{center}
\vspace{0mm}
\end{figure}

\noindent{\em Case III:} path $\<e_1,v_1,e_2,\dots, e_n,v_n\>$ belongs to tree $T$ and and $e_1\notin Cross(c)$. It follows from Definition~\ref{tree def} that path $\<v_1,e_2,\dots, e_n,v_n\>$ belongs to tree $T$ as well and that there is a unique $u_0\in in(e_1)\cap V_2$ such that path $\<u_0,e_1,v_1,e_2,\dots, e_n,v_n\>$ belongs to tree $T$. Hence, by definition~(\ref{f' def}),
\begin{eqnarray*}
&&
f'_{e_1}(\<e_1,v_1,e_2,\dots, e_n,v_n\>) +
\sum_{u\in in(e_1)}f'_u(\<u,e_1,v_1,e_2,\dots, e_n,v_n\>)\\
&=&
f'_{e_1}(\<e_1,v_1,e_2,\dots, e_n,v_n\>) +
f'_{u_0}(\<u_0,e_1,v_1,e_2,\dots, e_n,v_n\>) +\\
&&\sum_{u\in in(e_1)\setminus\{u_0\}}f'_u(\<u,e_1,v_1,e_2,\dots, e_n,v_n\>)\\
&&\hspace{5mm} \\
&=&
f_{e_1}(\<e_1,v_1,e_2,\dots, e_n,v_n\>) + 
(f_{u_0}(\<u_0,e_1,v_1,e_2,\dots, e_n,v_n\>)+ 1) +\\
&&\sum_{u\in in(e_1)\setminus\{u_0\}}f_u(\<u,e_1,v_1,e_2,\dots, e_n,v_n\>)\\
&&\hspace{5mm}  \\
&=&
f_{e_1}(\<e_1,v_1,e_2,\dots, e_n,v_n\>) + 1 + \\
&& \sum_{u\in in(e_1)}f_u(\<u,e_1,v_1,e_2,\dots, e_n,v_n\>)
 \pmod 2.
\end{eqnarray*}
Since $\<f_a\>_{a\in\mathcal{A}}\in\mathcal{L}$, by Definition~\ref{canonical L},
\begin{eqnarray*}
&& f_{e_1}(\<e_1,v_1,e_2,\dots, e_n,v_n\>) + \sum_{u\in in(e_1)}f_u(\<u,e_1,v_1,e_2,\dots, e_n,v_n\>) \nonumber\\
&=& f_{v_1}(\<v_1,e_2,\dots, e_n,v_n\>)\pmod 2.
\end{eqnarray*}
At the same time,
$f'_{v_1}(\<v_1,e_2,\dots, e_n,v_n\>) = f_{v_1}(\<v_1,e_2,\dots, e_n,v_n\>) + 1 \pmod 2$
by definition~(\ref{f' def}) since path $\<v_1,e_2,\dots, e_n,v_n\>$ belongs to tree $T$,
\begin{eqnarray*}
&& f'_{e_1}(\<e_1,v_1,e_2,\dots, e_n,v_n\>)+ 
\sum_{u\in in(e_1)}f'_u(\<u,e_1,v_1,e_2,\dots, e_n,v_n\>)\\
&=& f'_{v_1}(\<v_1,e_2,\dots, e_n,v_n\>) \pmod 2.
\end{eqnarray*}
This concludes the proof of the claim.
\end{proof}

\begin{claim}
$f'_u = f_u$ for each $u\in V_1$.
\end{claim}
\begin{proof}
Consider any vertex $u\in V_1$. Recall from Definition~\ref{Da def} that the domain of function $f_u$ is the set of all paths starting at vertex $u$. By Lemma~\ref{tree bounded lemma}, none of these paths belongs to tree $T$. Therefore, $f'_u = f_u$ due to definition~(\ref{f' def}).
\end{proof}

\begin{claim}
$f'_e = f_e$ for each $e\in E\setminus Cross(c)$.
\end{claim}
\begin{proof}
The statement of the claim follows from definition~(\ref{f' def}).
\end{proof}

\begin{claim}
$f'_b\neq f_b$.
\end{claim}
\begin{proof}
By Definition~\ref{tree def}, the single-element path $\<b\>$ belongs to tree $T$. Thus, $f'_b(\<b\>)=1+ f_b(\<b\>)$ due to definition~(\ref{f' def}). Therefore, $f'_b\neq f_b$.
\end{proof}

This concludes the proof of the lemma.
\end{proof}

\begin{lemma}\label{base case}
$H\Vdash A\rhd_p B$ if and only if $I_H\vDash A\rhd_p B$, for each $A,B\subseteq V$ and each non-negative real number $p$.
\end{lemma}
\begin{proof}
\noindent$(\Rightarrow)$. Suppose that $H\Vdash A\rhd_p B$. Then, by Definition~\ref{hyper sat}, there is a subset $F\subseteq E$ such that $w(F)\le p$ and $B\subseteq A^*_F$. By Definition~\ref{price info model}, inequality $w(F)\le p$ implies that $\|F\|\le p$. Thus, by Definition~\ref{info sat}, it suffices to show that for any two vectors $\ell_1,\ell_2\in\mathcal{L}$, if $\ell_1=_{A,F}\ell_2$, then $\ell_1=_{B}\ell_2$, which, in turn, follows from statement $B\subseteq A^*_F$ and Lemma~\ref{A*F lemma}.

\noindent$(\Leftarrow)$.  Assume that $I_H\vDash A\rhd_p B$. Thus, by Definition~\ref{info sat}, there is $F\subseteq \mathcal{A}=V\cup E$ such that $\|F\|\le p$ and for all $\ell_1,\ell_2\in\mathcal{L}$, if $\ell_1=_{A,F}\ell_2$, then $\ell_1=_B \ell_2$. Note that, by Definition~\ref{price info model}, $\|F\|\le p$ implies that $F\subseteq E$ and $w(F)\le p$. Suppose now that $H\nVdash A\rhd_p B$. Thus, $B\nsubseteq A^*_F$, by Definition~\ref{hyper sat}\idot Hence, there is $b\in B$ such that $b\notin A^*_F$. To finish the proof of the lemma, it suffices to construct $\ell_1,\ell_2\in\mathcal{L}$ such that $\ell_1=_{A,F}\ell_2$ and $\ell_1\neq_b \ell_2$. 
Consider cut $c=(A^*_F,E\setminus A^*_F)$. Let $\ell_1$ be vector $\mathbf{0}\in\mathcal{L}$. By Lemma~\ref{flip lemma}, there is vector $\ell_2$ such that $\ell_1=_{A^*_F,E\setminus Cross(c) }\ell_2$, and $\ell_1\neq_b \ell_2$. Note that $A\subseteq A^*_F$ by Definition~\ref{A*F} and Definition~\ref{AkF}\idot Thus, $\ell_1=_{A,E\setminus Cross(c) }\ell_2$. By Lemma~\ref{FCross}, we have $Cross(c)\cap F=\varnothing$. In other words, $F\subseteq E\setminus Cross(c)$. Therefore, $\ell_1=_{A,F}\ell_2$. 
\end{proof}
\begin{lemma}\label{induction lemma}
$H\Vdash\psi$ if and only if $I_H\vDash\psi$, for each $\phi\in\Phi(V)$.
\end{lemma}
\begin{proof}
We prove the lemma by induction on the structural complexity of formula $\psi$. The base case is shown in Lemma~\ref{base case}. The induction step follows from the induction hypothesis, Definition~\ref{hyper sat}, and Definition~\ref{info sat}\idot
\end{proof}

In the preceding part of this section, given any hypergraph $H$, we constructed a corresponding informational model $I_H$ and proved properties of this informational model. Next, we state and prove the completeness theorem for the informational semantics.  

\begin{theorem}\label{completeness}
For each set $\mathcal{A}$ and each formula $\phi\in\Phi(\mathcal{A})$, if $I\vDash\phi$ for each informational model $I=\<\mathcal{A},\{D_a\}_{a\in \mathcal{A}},\|\cdot\|,\mathcal{L}\>$, then $\vdash\phi$.
\end{theorem}
\begin{proof}
Assume that $\nvdash\phi$. By Theorem~\ref{hyper completeness}, there is a hypergraph $H=\<V,E,in,out,w\>$ such that  $\phi\in\Phi(V)$ and $H\nVdash\phi$.  Therefore, $I_H\nvDash\phi$ by Lemma~\ref{induction lemma}. 
\end{proof}

\section{Completeness theorem for the finite informational semantics}\label{finite info semantics completeness}

In the previous section, we have shown the completeness of our logical system with respect to the informational semantics. In Definition~\ref{finite model}, we introduced the notion of a finite informational model as a model in which each attribute has a finite cost. In this section we prove the completeness of our logical system with respect to the class of finite models. This is achieved by showing how any (potentially infinite) informational model could be converted to a finite model through the construction described below.

\begin{definition}\label{Ir}
For any non-negative real number $r$ and any informational model $I=\<\mathcal{A},\{D_a\}_{a\in \mathcal{A}},\|\cdot\|,\mathcal{L}\>$, let $I^r$ be tuple $\<\mathcal{A},\{D_a\}_{a\in \mathcal{A}},\|\cdot\|^r,\mathcal{L}\>$, where
$$
\|c\|^r=
\begin{cases}
\|c\|, & \mbox{ if $\|c\|\le r$},\\
r, & \mbox{otherwise}.
\end{cases}
$$
for each attribute $c\in\mathcal{A}$.
\end{definition}
\begin{corollary}
For any non-negative real number $r$ and any informational model $I$, tuple $I^r$ is a finite informational model. 
\end{corollary}

\begin{corollary}\label{le corollary}
$\|c\|^r\le \|c\|$, for each non-negative real number $r$, each informational model $I=\<\mathcal{A},\{D_a\}_{a\in \mathcal{A}},\|\cdot\|,\mathcal{L}\>$, and each attribute $c\in\mathcal{A}$.
\end{corollary}

\begin{definition}\label{rank}
For any $\phi\in\Phi(\mathcal{A})$, let  $rank(\phi)$ be defined recursively as follows:
\begin{enumerate}
\item $rank(A\rhd_p B)=p$,
\item $rank(\neg\psi)=rank(\psi)$,
\item $rank(\psi\to\chi)=max(rank(\psi),rank(\chi))$.
\end{enumerate}
\end{definition}

\begin{lemma}\label{rank lemma}
If $rank(\phi)< r$, then $I^r\vDash\phi$ if and only if $I\vDash\phi$.
\end{lemma}
\begin{proof}
We prove the lemma by induction on the structural complexity of formula $\phi$. The inductive step immediately follows from Definition~\ref{info sat}\idot Now, suppose that formula $\phi$ has form $A\rhd_p B$.

\noindent{$(\Rightarrow)$} If $I^r\vDash A\rhd_p B$, then, by Definition~\ref{info sat}, there is a set $C\subseteq \mathcal{A}$ such that (i) $\|C\|^r\le p$ and (ii) for each $\ell_1,\ell_2\in\mathcal{L}$ if $\ell_1=_{A,C}\ell_2$, then $\ell_1=_B\ell_2$. From condition (i), Definition~\ref{rank}, and assumption $rank(\phi)< r$ of the lemma,
\begin{equation}\label{Cr<r}
\|C\|^r\le p = rank(\phi)< r.
\end{equation}
Hence, $\|c\|^r\le \|C\|^r<r$ for each $c\in C$, by Definition~\ref{||set}\idot Thus, by Definition~\ref{Ir}, we have $\|c\|=\|c\|^r$. Then, by Definition~\ref{||set} and the first inequality in statement~(\ref{Cr<r}),
$$
\|C\|=\sum_{c\in C}\|c\|=\sum_{c\in C}\|c\|^r=\|C\|^r\le p.
$$
By Definition~\ref{info sat}, the inequality $\|C\|\le p$ together with condition (ii) above implies that $I\vDash A\rhd_p B$.

\noindent{$(\Leftarrow)$} If $I\vDash A\rhd_p B$, then, by Definition~\ref{info sat}, there is a set $C\subseteq \mathcal{A}$ such that (iii) $\|C\|\le p$ and (iv) for each $\ell_1,\ell_2\in\mathcal{L}$ if $\ell_1=_{A,C}\ell_2$, then $\ell_1=_B\ell_2$.
By Definition~\ref{||set}, Corollary~\ref{le corollary}, and again Definition~\ref{||set},
$$\|C\|^r=\sum_{c\in C}\|c\|^r\le \sum_{c\in C}\|c\| = \|C\|.$$
This along with condition (iii) implies that $\|C\|^r\le p$. Therefore, by Definition~\ref{info sat} and condition (iv), we have $I^r\vDash A\rhd_p B$.
\end{proof}

We next state and prove the completeness theorem for the class of finite informational models.

\begin{theorem}\label{finite completeness}
If $I\vDash\phi$ for each finite informational model $I=\<\mathcal{A},\{D_a\}_{a\in \mathcal{A}},\|\cdot\|,\mathcal{L}\>$ such that $\phi\in\Phi(\mathcal{A})$, then $\vdash\phi$.
\end{theorem}
\begin{proof}
Suppose that $\nvdash\phi$. Thus, by Theorem~\ref{completeness}, there is an informational model $I=\<\mathcal{A},\{D_a\}_{a\in \mathcal{A}},\|\cdot\|,\mathcal{L}\>$ such that $I\nvDash\phi$. Pick any $r$ such that $rank(\phi)<r$. Then, $I^r\nvDash\phi$ by Lemma~\ref{rank lemma}.
\end{proof}

\section{Conclusion}\label{conclusion section}

In this article we introduced a notion of the budget-constrained dependency that generalizes the notion of functional dependency previously studied  by~\cite{a74}. We propose a sound and complete axiomatization that captures the properties of the budget-constrained dependency. Although the axioms of our system are generalizations of the original Armstrong's axioms, the proof of the completeness for our system is significantly more complicated than its Armstrong's counterpart.  

\bibliography{sp}
\vspace*{10pt}

\end{document}